\documentclass[11pt]{article}
\usepackage{amsfonts}
\usepackage{amsmath}
\usepackage{revsymb4-1}
\usepackage{fullpage}
\usepackage{amssymb}
\usepackage{graphicx}
\usepackage{hyperref}
\usepackage{float}
\usepackage{soul}
\usepackage{bbding}
\usepackage{lscape}
\usepackage[normalem]{ulem}
\usepackage{stmaryrd}
\providecommand{\U}[1]{\protect\rule{.1in}{.1in}}

\newtheorem{theorem}{Theorem}

\newtheorem{corollary}[theorem]{Corollary}

\newtheorem{definition}[theorem]{Definition}

\newtheorem{lemma}[theorem]{Lemma}

\newtheorem{proposition}[theorem]{Proposition}
\newtheorem{remark}[theorem]{Remark}

\newenvironment{proof}[1][Proof]{\noindent\textbf{#1}\ \ }{\hfill\rule{0.5em}{0.5em}}

\def\Tr{\operatorname{Tr}}
\def\tr{\operatorname{Tr}}
\def\id{\operatorname{id}}
\def\str{\operatorname{str}}
\def\1{\openone}
\def\ox{\otimes}

\newcommand{\ket}[1]{|#1 \rangle}

\newcommand{\proj}[1]{|#1 \rangle \! \langle #1 |}
\newcommand{\SWAP}{\operatorname{SWAP}}

\let\originalleft\left
\let\originalright\right
\renewcommand{\left}{\mathopen{}\mathclose\bgroup\originalleft}
\renewcommand{\right}{\aftergroup\egroup\originalright}

\usepackage{color}

\begin{document}

\title{Capacities of Gaussian Quantum Channels \protect\\
       with Passive Environment Assistance}

\author{Samad Khabbazi Oskouei\footnote{{\tt kh{\_}oskuei@yahoo.com}; ${}^\circ${\tt stefano.mancini@unicam.it}; ${}^\times${\tt andreas.winter@uab.cat}}\\
\textit{{\small {Department of Mathematics}}}\\
\textit{{\small {Varamin-Pishva Branch}}}\\
\textit{{\small {Islamic Azad University}}}\\
\textit{{\small {Varamin, 33817-7489, Iran}}}\\
\and Stefano Mancini${}^\circ$\\
\textit{{\small {School of Science and Technology}}}\\
\textit{{\small {University of Camerino}}}\\
\textit{{\small {Via Madonna delle Carceri 9, 62032 Camerino, Italy}}}\\
\textit{{\small {\&{} INFN--Sezione Perugia}}}\\
\textit{{\small {Via A. Pascoli, 06123 Perugia, Italy}}}
\and Andreas Winter${}^\times$\\
\textit{{\small {ICREA--Instituci\'{o} Catalana de Recerca i Estudis Avan\c{c}ats}}}\\
\textit{{\small {Pg. Lluis Companys, 23, 08010 Barcelona, Spain}}}\\
\textit{{\small {\&{} Grup d'Informaci\'{o} Qu\`{a}ntica, Departament de F\'isica}}}\\
\textit{{\small {Universitat Aut\`{o}noma de Barcelona, 08193 Bellaterra (Barcelona), Spain}}}
}

\date{1 January 2021}

\maketitle

\thispagestyle{empty}

\begin{abstract}
Passive environment assisted communication takes place via a quantum channel
modeled as a unitary interaction between the information carrying system and
an environment, where the latter is controlled by a passive helper, who can
set its initial state such as to assist sender and receiver, but not help actively
by adjusting her behaviour depending on the message.
Here we investigate the information transmission capabilities in this framework by
considering Gaussian unitaries acting on Bosonic systems.

We consider both quantum communication and classical communication with helper,
as well as classical communication with free classical coordination between sender
and helper (conferencing encoders).

Concerning quantum communication, we prove general coding theorems with and
without energy constraints, yielding multi-letter (regularized) expressions.

In the search for cases where the capacity formula is computable, we look for
Gaussian unitaries that are universally degradable or anti-degradable. However, we
show that no Gaussian unitary yields either a degradable or anti-degradable
channel for all environment states. On the other hand, restricting to Gaussian
environment states, results in universally degradable unitaries, for which we
thus can give single-letter quantum capacity formulas.

Concerning classical communication, we prove a general coding theorem for the
classical capacity under and energy constraint, given by a multi-letter expression.
Furthermore, we derive an uncertainty-type relation between the classical capacities
of the sender and the helper, helped respectively by the other party, showing a
lower bound on the sum of the two capacities.
Then, this is used to lower bound the classical information transmission rate in
the scenario of classical communication between sender and helper.
\end{abstract}

\section{Introduction}
\label{sec:intro}
In quantum mechanics, every noisy channel (completely positive and
trace preserving -- cptp -- linear map) is the marginal of a reversible
(i.e.~unitary) interaction with an environment initially in a pure
state; this is the content of Stinespring's dilation theorem
\cite{Stinespring}, and of the subsequent structure theorems of
Choi \cite{Choi}, Jamio{\l}kowski \cite{Jam} and Kraus \cite{Kraus}.
This feature, which distinguishes quantum communication fundamentally
from its classical counterpart, is at the bottom of the possibility to
perform unconditional secret key agreement over a channel, since the
channel essentially uniquely determines the action on the environment.
In this picture, noise in the channel is entirely due to loss of information
into the environment, more precisely the build-up of correlations between
the system and the environment. A series of prior work, starting with
\cite{GW:lostnfound1,GW:lostnfound2} have asked how much one can counteract
the noise if one had access to the environment output state and could feed classical
information back into the channel output system \cite{HaydenKing,SVW,W:lostnfound,MCM,MMM}.

Somewhat dually, two of the present authors have asked previously,
of what benefit can be access to the \emph{initial} state of the
environment \cite{KSWY,KSWY2}. In contrast to the (active) interventions
in the environment of the aforementioned works, we call this passive
environment assistance, since the role of the helper is restricted
to choosing a suitable initial state.
These previous results were obtained in the finite-dimensional setting.
Here, we extend the model and results to infinite-dimensional systems,
with special attention to Gaussian channels and their Gaussian unitary
dilations. Additional motivations for the model of passive environment assistance
comes  from the notion of \emph{environment-parametrized quantum channels}, which
are used to describe quantum memory cells \cite{DW19}.

\medskip
The present paper is structured as follows:
In Section \ref{sec:models} we define the system mode and establish
basic notation. In Section \ref{sec:quantum} we treat quantum communication
capacities both without and with energy constraints; we show that
two-mode Gaussian unitaries are never universally degradable or
anti-degradable, but restricting to Gaussian helper there are
families of either type, allowing us to explicitly calculate
the passive environment-assisted quantum capacity under this restriction.
In Section \ref{sec:classical} we analyze the classical capacity with
a helper under energy constraints both for sender and helper; we show
that the capacity of the sender assisted by the helper and of the
helper assisted by the sender cannot both be small, and apply this
insight to the case of conferencing encoder.
In Section \ref{sec:continuity} we prove that the passive environment-assisted
quantum and classical capacities with energy constraints are continuous
in the unitary interaction, and indeed uniformly so with respect to the
energy-constrained diamond norm. In Section \ref{sec:conclusion} we conclude.
Two appendices provide additional proofs: In Appendix \ref{app:thmproof} we prove
Theorem \ref{thm:nongaussian}, stating that non-trivial two-mode Gaussian unitaries
are neither universally degradable nor universally anti-degradable;
Appendix \ref{app:OMB} proves tighter lower bounds on the sum of classical
capacities for two-mode Gaussian unitaries.


\section{System model and notation}
\label{sec:models}

Let ${\cal L}(X)$ denote the space of linear operators on a (separable) Hilbert space $X$.
We denote the identity operator in ${\cal L}(X)$ as $\1_X$ and the identity map
(ideal channel) $\id : {\cal L}(X) \to {\cal L}(X)$ is denoted by $\id_X$.
For any linear operator $\Lambda : A \to B$ between Hilbert spaces we
denote the trace norm
\begin{equation}
  \|\Lambda\|_1 := \tr\sqrt{\Lambda^\dag\Lambda} = \tr|\Lambda|,
\end{equation}
and the operator norm
\begin{equation}
 \|\Lambda\|_\infty := \sup \left\{ \bigl|\Lambda \ket{\psi}\bigr|
                                     : \ket{\psi} \in A, \bigl|\ket{\psi}\bigr|=1 \right\},
\end{equation}
where $|\cdot|$ denotes the Hilbert space norm.
Let ${\cal T}(X)\subset {\cal L}(X)$ denote the set of trace class operators
whose trace norm, defined above, is finite; likewise, ${\cal B}(X)\subset {\cal L}(X)$
is the set of bounded operators, whose operator norm is finite.
Any positive semidefinite element $\rho\in {\cal T}(X)$ with $\tr\rho=1$ is called a
density operator. Obviously the set ${\cal S}(X)$ of such operators is a proper subset of
${\cal T}(X)$.
A quantum channel $\mathcal{N}$ from system $A$ to a system $B$ is a completely
positive and trace preserving  (CPTP) linear map  from  ${\cal T}(A)$ to ${\cal T}(B)$.
Furthermore, a linear map $\mathcal{N}: {\cal T}(A) \to {\cal T}(B)$ is called Hermitian
preserving if for any bounded Hermitian operator $O$, also $\mathcal{N}(O)$ results Hermitian.

For a density operator $\alpha$, the von Neumann entropy is defined as
\begin{equation}\label{eq:SvonN}
  S(\alpha) := -\tr \alpha\ln\alpha .
\end{equation}
Throughout the paper we use natural logarithms $\ln$, as is customary
in settings of continuous alphabets, resulting in the entropy and capacity
be counted in units of \emph{nats}.
For two density operators $\alpha$ and $\beta$
such that $supp(\alpha)\subseteq supp(\beta)$, the quantum relative entropy of
$\alpha$ with respect to $\beta$ is defined as
\begin{equation}\label{eq:qrelent}
  D(\alpha\|\beta) := \tr \alpha (\ln\alpha-\ln\beta);
\end{equation}
otherwise, $D(\alpha\|\beta) := \infty$.

A Hamiltonian $H_A$ is a densely defined self-adjoint operator on the Hilbert
space of a quantum system $A$, that is bounded from below.
One way of defining such an operator is
to let $\{|e_j\rangle\}$ be an orthonormal basis for the Hilbert
space under consideration (e.g. Fock basis), and $\{a_j\}$, a sequence of real numbers bounded from below.
Then,
\begin{equation}
  H_A |\psi\rangle := \sum_{j=1}^{\infty} a_j |e_j\rangle\langle e_j |\psi\rangle,
\end{equation}
defines $H_A$ on the dense subspace
$\mathcal{I}=\left\{ |\psi\rangle : \sum_{j=1}^\infty a_j^2 |\langle e_j |\psi\rangle|^2 < +\infty\right\}$,
with $\{a_j\}$ the eigenvalues corresponding to the eigenvectors $\{|e_j\rangle\}$.
All Hamiltonians with discrete spectrum arise in this way.

For an arbitrary state $\rho$, the expectation of $H_A$ is given by
\begin{equation}
  \tr \rho H_A = \sum_{j=1}^\infty a_j \langle e_j|\rho|e_j\rangle.
\end{equation}
The $n$-th extension $H_{A^n}$ of the energy observable $H_A$ to
the system $A^n=A^{\otimes n}$ is defined in an i.i.d. fashion as follows:
\begin{equation}
  H_{A^n}:= H_A\ox \1\ox \cdots \ox \1
            +\1 \ox H_A\ox \cdots \ox \1
            + \ldots + \1\ox\cdots\ox\1\ox H_A .
\end{equation}

In the present paper, we
consider a communication model between Alice and Bob that involves
also a third party (helper) controlling the environment input system,
whose aim is to enhance the communication between Alice and Bob.
We assume that the helper sets the initial state of the environment to
enhance the communication from Alice to Bob, and then has no role in the
coding protocol
(thus we refer to this model as \emph{passive environment-assisted model}).

Consider an isometry $W : A \otimes E \to B \otimes F$ which defines a channel
${\cal N} : {\cal L}(A \otimes E) \to {\cal L}(B)$, whose action
on the input state $\sigma$ on $A \otimes E$ is
\begin{equation}
  \label{eq:Ntot}
  {\cal N}^{AE\to B}(\sigma) = \tr_F W\sigma W^\dag .
\end{equation}
Then an effective channel ${\cal N}_\eta : {\cal L}(A) \to {\cal L}(B)$ is
established between Alice and Bob once the initial state $\eta$ on $E$ is set:
\begin{equation}
  \label{eq:Neta}
  {\cal N}^{A\to B}_\eta (\rho) := {\cal N}^{AE\to B}(\rho \otimes \eta).
\end{equation}
The complementary channel is
\begin{equation}
  \label{eq:Ntildeeta}
  {\widetilde{\cal N}}^{A\to F}_\eta (\rho) :=\tr_B W (\rho\otimes \eta) W^\dag,
\end{equation}
while the adjoint channel ${\cal N}^{* \, B\to A}_\eta$ acts on the
bounded operator $b \in {\cal B}(B)$ such that
\begin{equation}
  \label{eq:Nstareta}
  \tr\left[  {\cal N}^{* \, B\to A}_\eta (b) \, \rho\right]
     = \tr \left[ b \, {\cal N}^{A\to B}_\eta(\rho)\right] \qquad b\in {\cal B}(A).
\end{equation}
It can be written in the explicit form
\begin{equation}
  \mathcal{N}^{* \, B\to A}_\eta(b) = \tr_E W^\dag (b\otimes\1_F) W (\1_A\otimes\eta),
\end{equation}
using the isometry $W$ and the state $\eta$.


\subsection{Gaussian states}

Let us recall some basic facts about Gaussian states,
which also serves the purpose of fixing the notations used in following
sections. The canonical observables
$\hat{\boldsymbol{r}} =\left(\hat q_1,\hat p_1,\ldots, \hat q_N, \hat p_N \right)^\top$
describe a Bosonic system of $N$ harmonic modes in a Hilbert space
$X = \bigotimes_{k=1}^N X_k$.
On such a system, we consider by default a quadratic Hamiltonian,
whose most general form is
\begin{equation}\label{eq:HX}
  H_X = \hat{\boldsymbol{r}}_X \boldsymbol{\Omega}_X \hat{\boldsymbol{r}}_X^\top,
\end{equation}
where $\boldsymbol{\Omega}_X$ is positive symmetric matrix,
assumed for the sake of simplicity to have a unique $N$-fold degenerate
eigenvalue $\omega^X$. Hence in normal form,
$H_X = \omega^X \sum_{j=1}^{N} (\hat{q}_j^2 + \hat{p}_j^2)/2$.

Hereafter we denote vectors (resp. matrices) by lower (resp. upper) case bold symbols.
The Heisenberg canonical commutation relations satisfied by the canonical observables
can be compactly represented as
\begin{equation}\label{eq:13}
\left[\hat r_j,\hat r_k\right]=i\Sigma_{jk},
\quad  \forall j,k\in\{1,\ldots,2N\},
\end{equation}
with
\begin{equation}
\boldsymbol{\Sigma}:=\bigoplus_1^N \begin{pmatrix}
                                                    0 & 1 \\
                                                    -1 & 0
                                                  \end{pmatrix},
\end{equation}
and $\hat r_{2k-1}=\hat q_k$, $\hat r_{2k}=\hat p_k$.
For any density operator $\rho$ acting on $X$, the vector mean
(or first moment) is the vector ${\boldsymbol{d}}\in \mathbb{R}^{2N}$,
whose components are given by
\begin{equation}
  d_k:=\tr \rho \hat r_k .
\end{equation}
The $2N\times 2N$ covariance matrix (CM) $\boldsymbol V$ is given by
\begin{equation}
  {V}_{jk}:=\tr \rho
                \left\{\left(\hat r_j-d_j\right)\left(\hat r_k-d_k\right)
                      + \left(\hat r_k-d_k\right)\left(\hat r_j-d_j\right) \right\},
\end{equation}
which is real, symmetric and positive definite.
Furthermore, for the CM to correspond a bona fide quantum
state it has to satisfies the following Heisenberg-Robertson uncertainty relation
\begin{equation}\label{eq:UR}
  {\boldsymbol V}+i{\boldsymbol{\Sigma}} \geq 0.
\end{equation}
Conversely, if the uncertainty relation is satisfied, there exists a quantum
state with CM ${\boldsymbol V}$, in fact a Gaussian state $\rho$. It is uniquely defined
by its associated a Gaussian characteristic function
\begin{equation}\label{eq:20}
  \chi_\rho({\boldsymbol\zeta})
    = \exp\left(-i \left(\boldsymbol{\Sigma} {\boldsymbol{ d}}\right)^\top {\boldsymbol\zeta}
                -\frac{1}{4} {\boldsymbol\zeta}^\top \boldsymbol{\Sigma}
                             {\boldsymbol{V}} \boldsymbol{\Sigma}^\top{\boldsymbol\zeta}
           \right),
\end{equation}
where ${\boldsymbol\zeta}\in\mathbb{R}^{2N}$.
Recall that the (zero-ordered) characteristic function is defined as
\begin{equation}
  \chi_\rho({\boldsymbol\zeta}):=\tr\left(\rho W_{\boldsymbol\zeta}\right),
\end{equation}
with the Weyl displacement operator given by
\begin{equation}
  W_{\boldsymbol\zeta}:=\exp\left(-i{\hat{\boldsymbol{r}}}^\top\boldsymbol{\Sigma} {\boldsymbol\zeta}\right).
\end{equation}

Thus Gaussian states are completely characterized by $\boldsymbol{d}$ and $\boldsymbol{V}$.

The von Neumann entropy \eqref{eq:SvonN} of an $N$-mode Gaussian state $\rho$ can be evaluated
through its covariance matrix as
\begin{equation}
  S(\rho) = S(\boldsymbol{V}) = \sum_{i=1}^{N} g(\nu_i),
\end{equation}
where $\nu_1, \ldots, \nu_N$ are the symplectic eigenvalues of $\mathbf{V}$.
Note that for Gaussian states, the entropy is a function entirely
of the CM, and so we slightly abuse notation writing $S(\boldsymbol{V})$.
Here the function $g$ is defined by
\begin{equation}
  \label{eq:gfunction}
  g(x) := \left(x+\frac{1}{2}\right)\ln\left(x+\frac{1}{2}\right)
                 -\left(x-\frac{1}{2}\right)\ln\left(x-\frac{1}{2}\right),
\end{equation}
and as such $g(x)$ is an increasing and concave function.


\subsection{Gaussian unitaries}\label{subsec:GU}

Consider $N$ Bosonic modes.
A Gaussian unitary on them $\exp(-iH)$ with $H$ as in Eq.~\eqref{eq:HX},
can be simply described by an affine map
\begin{equation}
  (\boldsymbol{S},\boldsymbol{\zeta}) :\hat{\boldsymbol{r}}\to \boldsymbol{S}\, \hat{\boldsymbol{r}} + \boldsymbol{\zeta},
\end{equation}
where $\boldsymbol{\zeta}\in \mathbb{R}^{2 N}$ and $\boldsymbol{S}\in Sp(2N,\mathbb{R})$ because the  transformation must preserve the commutation relations \eqref{eq:13}.
Clearly the eigenvalues $\boldsymbol{r}$ of the quadrature operators $\hat{\boldsymbol{r}}$ must follow the same rule, i.e.,
\begin{equation}
  (\boldsymbol{S},\boldsymbol{\zeta}) :\boldsymbol{r}\to \boldsymbol{S}\, \boldsymbol{r}+ \boldsymbol{\zeta}.
\end{equation}
Thus, a Gaussian unitary is equivalent to an affine symplectic map $(\boldsymbol{S},\boldsymbol{\zeta})$ acting on the phase space, and can be denoted by $\boldsymbol{U}_{\boldsymbol{S}, \boldsymbol{\zeta}}$. In particular, we can write
 \begin{equation}\label{eq:Ucanonical}
 \boldsymbol{U}_{\boldsymbol{S},\boldsymbol{\zeta}}=W_{\boldsymbol{\zeta}}
 \boldsymbol{U}_{\boldsymbol{S}},
\end{equation}
where the canonical unitary $\boldsymbol{U}_{\boldsymbol{S}}$ corresponds to a linear
symplectic map  $\boldsymbol{r}\to \boldsymbol{S}\, \boldsymbol{r}$, and the Weyl operator
$W_{\boldsymbol{\zeta}}$ to a phase-space
translation $\boldsymbol{r}\to \boldsymbol{r}+ \boldsymbol{\zeta}$.

In terms of the statistical moments, ${\boldsymbol{d}}$ and $\boldsymbol{V}$, the action of
$\boldsymbol{U}_{\boldsymbol{S}, \boldsymbol{\zeta}}$ is characterized by the following transformations
\begin{equation} \label{eq:22-1}
 {\boldsymbol{d}}\to \boldsymbol{S} {\boldsymbol{d}}+ \boldsymbol{\zeta},\quad{\boldsymbol{V}}\to \boldsymbol{S} \boldsymbol{V} \boldsymbol{S}^\top.
\end{equation}
Therefore, the action of a Gaussian unitary $\boldsymbol{U}_{\boldsymbol{S}, \boldsymbol{\zeta}}$
over a Gaussian state $\rho({\boldsymbol{d}}, \boldsymbol{V})$ will be completely described by Eq.~\eqref{eq:22-1}.

Note that the above arguments also apply if we replace the vector of quadrature operators
$\hat{\boldsymbol{r}}$ by the vector of ladder operators (also known as annihilation
and creation operators)
$\hat{\boldsymbol{\upsilon}}=(\hat{a}_1, \hat{a}_1^\dagger, \cdots, \hat{a}_n, \hat{a}_n^\dagger)^\top$,
where
\begin{equation}
\hat{a}_j=\frac{\hat{q}_j+i\hat{p}_j}{\sqrt{2}}.
\end{equation}

In such a case however, it will be $\boldsymbol{S}\in Sp(2N,\mathbb{C})$.
Let us now focus on two-mode Gaussian unitaries. Consider
$\hat{\boldsymbol{\upsilon}}=(\hat{a}, \hat{a}^\dagger, \hat{b}, \hat{b}^\dagger)^\top$
with
\begin{equation}
\hat{a}=\frac{\hat{q}_a+i\hat{p}_a}{\sqrt{2}},
\quad
\hat{b}=\frac{\hat{q}_b+i\hat{p}_b}{\sqrt{2}}.
\end{equation}
Then, the canonical unitary of Eq.~\eqref{eq:Ucanonical}, named here ${\boldsymbol U}_{ab}$,
satisfies
\begin{equation}
 {\boldsymbol U}_{ab} \, \hat{\boldsymbol{\upsilon}} \, {\boldsymbol U}_{ab}^\dagger=  {\boldsymbol  S}\cdot \hat{\boldsymbol{\upsilon}},
\end{equation}
with ${\boldsymbol S}\in Sp(4,\mathbb{C})$.
Define
\begin{equation}
 q=\vert S_{11}\vert^2 -\vert S_{12}\vert^2 ,
 \end{equation}
where $S_{11}$ and $S_{12}$ are matrix elements of  $ {\boldsymbol S}$.
In \cite[App.~A]{CG06}, it is shown that for $0<q, q\neq 1$
\begin{equation}
 {\boldsymbol U}_{ab}= \left( {\boldsymbol S}_a \otimes  {\boldsymbol S}_b\right)
  {\boldsymbol U}_{ab}^{(q)}  \left(I_a\otimes  {\boldsymbol S}'_b\right),
\end{equation}
where $ {\boldsymbol S}_a$, $ {\boldsymbol S}_b$ and $ {\boldsymbol S}'_b$ are one-mode squeezing transformations.
For $q\in(0, 1)$, $ {\boldsymbol U}_{ab}^{(q)}$ is characterized by the symplectic matrix
\begin{equation}\label{Scan1}
 {\boldsymbol S}_{ab}^{(q)}=\left(\begin{array}{cccc}
\sqrt{q} & 0 & -\sqrt{1-q} & 0  \\
0 & \sqrt{q} & 0 & -\sqrt{1-q} \\
\sqrt{1-q} & 0 & \sqrt{q} & 0 \\
0 & \sqrt{1-q} & 0 & \sqrt{q}
\end{array}\right),
\end{equation}
while for $q>1$, by
\begin{equation}\label{Scan2}
 {\boldsymbol S}_{ab}^{(q)}=\left(\begin{array}{cccc}
\sqrt{q} & 0 & 0 & -\sqrt{q-1}  \\
0 & \sqrt{q} & -\sqrt{q-1} & 0 \\
0 & -\sqrt{q-1} & \sqrt{q} & 0 \\
-\sqrt{q-1} & 0 & 0 & \sqrt{q}
\end{array}\right).
\end{equation}
The case $q<0$ can be traced back to the case $q>0$ by the following argument.
Consider the transformation $\SWAP_{ab}$ swapping (exchanging) the two modes, defined by
\begin{equation}
 \SWAP_{ab} = \SWAP_{ab}^\dagger,\quad
 \SWAP_{ab} \hat{a} \SWAP_{ab}^\dagger=\hat{b}, \quad
 \SWAP_{ab} \hat{b} \SWAP_{ab}^\dagger=\hat{a}.
\end{equation}
Therefore, one gets the following relation:
\begin{equation}
 {\boldsymbol\SWAP}_{ab}  {\boldsymbol U}_{ab} \hat{\boldsymbol{\upsilon}}
 {\boldsymbol U}_{ab}^\dagger  {\boldsymbol\SWAP}_{ab}
 = \widetilde{{\boldsymbol S}}\cdot \hat{\boldsymbol{\upsilon}},
\end{equation}
where $\widetilde{ {\boldsymbol S}}$ is a $4\times 4$ matrix obtained by shifting by $2$
the columns of the symplectic matrix ${\boldsymbol S}$ describing the unitary ${\boldsymbol U}_{ab}$.
In other words,
\begin{equation}
  \widetilde{S}_{ij}=S_{i, j\oplus 2} ,
\end{equation}
where $\oplus$ denotes the sum modulo $4$. In this way we have
\begin{equation}
 {\boldsymbol U}_{ab}=   {\boldsymbol\SWAP}_{ab} \left( {\boldsymbol S}_a \otimes  {\boldsymbol S}_b\right)
  {\boldsymbol U}_{ab}^{(1-q)}  {\boldsymbol S}'_b.
\end{equation}


\subsection{Gaussian quantum channels}

A Bosonic Gaussian channel (BGC) ${\cal N}^{A\to B}$ is a linear
completely positive and trace preserving map defined on ${\cal T}(A)$
and taking values in ${\cal T}(B)$, that maps every Gaussian state
to a Gaussian state.
As Gaussian states span all states and are completely characterized by
their first and second moments, the BGC ${\cal N}^{A\to B}$ can be completely
characterized by the rule of transformations on the vector mean and
the covariance matrix. On the level of vector mean and covariance matrices,
the action of ${\cal N}^{A\to B}$ is as follows:
\begin{equation}
  \label{dVtransforms}
  \begin{split}
  \boldsymbol{d}_A\mapsto \boldsymbol{d}_B
                  &= \boldsymbol{X} \boldsymbol{d}_A + \boldsymbol{d}_E,\\
  \boldsymbol{V}_A\mapsto \boldsymbol{V}_B
                  &={\boldsymbol{ X}}{\boldsymbol{ V}}_A{\boldsymbol{ X}}^\top+{\boldsymbol{ Y}},
  \end{split}
\end{equation}
where ${\boldsymbol{X}}$ and ${\boldsymbol{Y}}$ are real matrices with
${\boldsymbol{ Y}} = {\boldsymbol{ Y}}^\top$ and ${\boldsymbol{ Y}}\geq 0$.
For this transformation to represent a bona fide quantum channel, in other
words taking into account the complete positivity condition, we must have
\begin{equation}
  {\boldsymbol{ Y}}+i\boldsymbol{\Sigma} \geq i {\boldsymbol{ X}}\boldsymbol{\Sigma}{\boldsymbol{X}}^\top.
\end{equation}
In particular, when ${\boldsymbol{ Y}} = 0$, the channel ${\cal N}^{A\to B}$
represents a unitary evolution of the system and from Eq. \eqref{eq:UR},
it follows that $\boldsymbol{ X}$ is a symplectic matrix. Thus, the action
of a Gaussian unitary $U^{A\to B}$ on the state $\rho^A$ with $N_A$ modes can
be described by a symplectic matrix of size $2N_A\times 2N_A$ as follows:
\begin{equation}
  \rho^B = U\rho^A U^\dag \leftrightarrow {\boldsymbol{ V}}_B
         = {\boldsymbol{ S}}{\boldsymbol{ V}}_A {\boldsymbol{ S}}^\top.
\end{equation}
It is furthermore well-known that a quantum channel can be seen as part of a
unitary evolution on a larger system whose ancillary parts are not under our control.
Actually, every BGC acting on $N_A$ modes can be represented by a unitary operation
$U^{AE\to BF}$ on the system and a minimal environment of $N_E$ modes, where
$N_E\leq 2N_A$. This unitary interaction,
extending the argument of Subsection \ref{subsec:GU}) to multimodes,
 can be described by a symplectic matrix
$\boldsymbol{S}$ , written in block form as follows:
\begin{equation}
  \label{sympl}
  \boldsymbol{S} = \begin{pmatrix}
                     {\boldsymbol{ M}} & {\boldsymbol{ N}} \\
                     {\boldsymbol{ O}} & {\boldsymbol{ P}}
                   \end{pmatrix}.
\end{equation}
When the input state in the environment is ${\boldsymbol{ V}}_E$, the effective
channel ${\cal N}^{A\to B}_{\boldsymbol{ V}_E}$ can be described as
\begin{equation}\label{eq:22}
  {\boldsymbol{V}}_A \mapsto {\boldsymbol{ V}}_B
                     =       {\boldsymbol{ M}} {\boldsymbol{ V}}_A {\boldsymbol{ M}}^\top
                             + {\boldsymbol{ N}} {\boldsymbol{ V}}_E {\boldsymbol{ N}}^\top.
\end{equation}
In turn, the complementary channel $\widetilde{\cal N}^{A\to F}_{\boldsymbol{ V}_E}$
acts on the CM as
\begin{equation}\label{eq:23}
  {\boldsymbol{ V}}_A \mapsto {\boldsymbol{ V}}_F
                      =       {\boldsymbol{ O}}{\boldsymbol{ V}}_A {\boldsymbol{ O}}^\top
                               +{\boldsymbol{ P}} {\boldsymbol{ V}}_E {\boldsymbol{ P}}^\top.
\end{equation}

\begin{lemma}
\label{lem:lemma2}
Let $\mathcal{N}^{AE\rightarrow B}$ be a Gaussian channel from system $AE$ to system
$B$ with input Gaussian states subject to the conditions $\Tr \rho H_A \leq P_A$
and $\Tr \eta H_E \leq P_E$, for density operators  $\rho$ and $\eta$ on systems $A$ and $E$,
respectively.
Then, there exists a quadratic Hamiltonian $H_B$ on system $B$ such that
\begin{equation}
  \label{eq:cohamiltonian}
  \Tr \mathcal{N}(\rho\ox\eta) H_B \leq 2P_A+ 2P_E, \quad\text{and}\quad \Tr e^{-\beta H_B}< \infty , \ \forall \beta>0.
\end{equation}
Furthermore, it holds
\begin{equation}\label{eq:coentropy}
 \sup_{\eta : \Tr\left(\eta  H_{E} \right)\leq P_E}\,\, \sup_{\rho : \Tr\left(\rho (H_{A}) \right)\leq P_A}
                         S(\mathcal{N}(\rho\ox\eta)) < \infty.
\end{equation}
\end{lemma}

\begin{proof}
Let us generically consider each system $A$, $E$, $B$ to be composed of $N$ modes,
and recall from Eq.~\eqref{eq:HX}, that
\begin{equation}
  H_A = \hat{\boldsymbol{r}}_A \boldsymbol{\Omega}_A \hat{\boldsymbol{r}}_A^\top,
\end{equation}
as well as
\begin{equation}
   H_E=\hat{\boldsymbol{r}}_E \boldsymbol{\Omega}_E \hat{\boldsymbol{r}}_E^\top,
\end{equation}
to be quadratic Hamiltonians, where $\boldsymbol{\Omega}_A$ and $\boldsymbol{\Omega}_E$
are positive matrices with eigenvalues $\omega^A$ and $\omega^E$.

On the system $A$ (resp. $E$), for a given state $\rho$ (resp. $\eta$) with covariance matrix
$\boldsymbol{V}_\rho$ (resp. $\boldsymbol{V}_\eta$) the constrained energy is given by
$\Tr \rho H_A
  =\Tr \boldsymbol{\Omega}_A \boldsymbol{V}_\rho + \boldsymbol{d}_A \boldsymbol{\Omega}_A \boldsymbol{d}^\top_A
  \leq P_{A}$,
and similarly
$\Tr \rho H_E
  =\Tr \boldsymbol{\Omega}_E \boldsymbol{V}_\eta + \boldsymbol{d}_E \boldsymbol{\Omega}_E \boldsymbol{d}^\top_E
  \leq P_{E}$.
Let us define
\begin{equation}\label{eq:conHB}
  H_{B, \eta} := c \left(\hat{\boldsymbol{r}}_B \hat{\boldsymbol{r}}_B^\top - (\Tr \boldsymbol{V}_\eta \boldsymbol{N}^\top \boldsymbol{N}) \1_B - \boldsymbol{d}_\eta  \boldsymbol{N}^\top \boldsymbol{N} \boldsymbol{d}_\eta^\top \1_B \right),
\end{equation}
where $c$ is a positive real constant.
We know that $\boldsymbol{N}$ and $\boldsymbol{\Omega}_E$ are finite dimensional matrices.
Therefore, it is possible to choose a constant $c_E>0$ such that
$\boldsymbol{N}^\top \boldsymbol{N}\leq c_E \boldsymbol{\Omega}_E$.
As a consequence,
\begin{equation}
   \Tr \boldsymbol{V}_\eta \boldsymbol{N}^\top \boldsymbol{N}
      + \boldsymbol{d}_\eta  \boldsymbol{N}^\top \boldsymbol{N} \boldsymbol{d}_\eta^\top
   \leq c_E\Tr(\boldsymbol{\Omega}_E \boldsymbol{V}_\eta)
      + c_E \boldsymbol{d}_\eta  \boldsymbol{\Omega}_E \boldsymbol{d}_\eta^\top
    = c_E\Tr \eta H_E \leq c_E P_{E},
\end{equation}
hence
\begin{equation}
      H_{B, \eta} \geq c(\hat{\boldsymbol{r}}_B\hat{\boldsymbol{r}}_B^\top -c_E P_{E}\1_B).
\end{equation}
In other words, the eigenvalues of $H_{B, \eta}$ are bounded from below. Therefore, we have
\begin{align}\label{re:gibes}
 \Tr \exp(-\beta H_{B, \eta})
   &\leq \Tr \exp(-\beta c\hat{\boldsymbol{r}}_B \hat{\boldsymbol{r}}_B^\top)
          \exp(\beta c c_E P_{E}) < \infty.
\end{align}
On the other hand, we have
\begin{equation}
  \Tr \rho \mathcal{N}_\eta^*(H_{B,\eta})
    = \Tr \mathcal{N}_\eta(\rho) H_{B, \eta}
    = c\Tr \mathcal{N}_\eta(\rho) \hat{\boldsymbol{r}}_B \hat{\boldsymbol{r}}_B^\top
      - c\Tr \boldsymbol{V}_\eta \boldsymbol{N}^\top \boldsymbol{N}
      - c\boldsymbol{d}_\eta  \boldsymbol{N}^\top \boldsymbol{N} \boldsymbol{d}_\eta^\top,
\end{equation}
hence
\begin{equation}\begin{split}
\Tr \rho \mathcal{N}_\eta^*(H_{B, \eta})
  &= c\Tr(\boldsymbol{M}^\top \boldsymbol{V}_\rho \boldsymbol{M} + \boldsymbol{N}^\top \boldsymbol{V}_\eta \boldsymbol{N})
      + c\left(\boldsymbol{d}_\rho \boldsymbol{M}^\top+\boldsymbol{d}_\eta \boldsymbol{N}^\top\right)
         \left(\boldsymbol{d}_\rho \boldsymbol{M}^\top+\boldsymbol{d}_\eta \boldsymbol{N}^\top\right)^\top \\
  &\phantom{=}
   - c\Tr \boldsymbol{V}_\eta \boldsymbol{N}^\top \boldsymbol{N}
   - c\boldsymbol{d}_\eta  \boldsymbol{N}^\top \boldsymbol{N} \boldsymbol{d}_\eta^\top \\
  &= c\Tr \boldsymbol{M}^\top \boldsymbol{V}_\rho \boldsymbol{M}
     + c\left(\boldsymbol{d}_\rho \boldsymbol{M}^\top
             +\boldsymbol{d}_\eta \boldsymbol{N}^\top\right)
        \left(\boldsymbol{d}_\rho \boldsymbol{M}^\top+\boldsymbol{d}_\eta \boldsymbol{N}^\top\right)^\top
     - c\boldsymbol{d}_\eta  \boldsymbol{N}^\top \boldsymbol{N} \boldsymbol{d}_\eta^\top.
\end{split}\end{equation}
From the triangle inequality, we have
\begin{equation}
 \left(\boldsymbol{d}_\rho \boldsymbol{M}^\top+\boldsymbol{d}_\eta \boldsymbol{N}^\top\right) \left(\boldsymbol{d}_\rho \boldsymbol{M}^\top+\boldsymbol{d}_\eta \boldsymbol{N}^\top\right)^\top\leq  2\boldsymbol{d}_\rho \boldsymbol{M}^\top \boldsymbol{M}\boldsymbol{d}_\rho^\top   +2\boldsymbol{d}_\eta \boldsymbol{N}^\top\boldsymbol{N}\boldsymbol{d}_\eta^\top.
\end{equation}
From the above inequality, one gets
\begin{equation}\begin{split}
  \Tr \rho \mathcal{N}_\eta^*(H_{B, \eta})
    &\leq c\Tr \rho \, \hat{\boldsymbol{r}}_A \boldsymbol{M}^\top \boldsymbol{M}\hat{\boldsymbol{r}}_A^\top
          + 2c\boldsymbol{d}_\rho \boldsymbol{M}^\top \boldsymbol{M}\boldsymbol{d}_\rho^\top
          + 2c\boldsymbol{d}_\eta \boldsymbol{N}^\top\boldsymbol{N}\boldsymbol{d}_\eta^\top
          - c\boldsymbol{d}_\eta  \boldsymbol{N}^\top \boldsymbol{N} \boldsymbol{d}_\eta^\top \\
    &\leq c\Tr \rho \, \hat{\boldsymbol{r}}_A \boldsymbol{M}^\top \boldsymbol{M}\hat{\boldsymbol{r}}_A^\top
          + 2c\boldsymbol{d}_\rho \boldsymbol{M}^\top \boldsymbol{M}\boldsymbol{d}_\rho^\top
          + c\boldsymbol{d}_\eta \boldsymbol{N}^\top\boldsymbol{N}\boldsymbol{d}_\eta^\top  \\
    &\leq c\Tr \rho \, \hat{\boldsymbol{r}}_A \boldsymbol{M}^\top \boldsymbol{M}\hat{\boldsymbol{r}}_A^\top
          + 2c\boldsymbol{d}_\rho \boldsymbol{M}^\top \boldsymbol{M}\boldsymbol{d}_\rho^\top
          + c c_E\boldsymbol{d}_\eta    \boldsymbol{\Omega}_E   \boldsymbol{d}_\eta^\top \\
    &\leq \Tr \rho \, \hat{\boldsymbol{r}}_A \boldsymbol{\Omega}_A\hat{\boldsymbol{r}}_A^\top
           + 2\boldsymbol{d}_\rho \boldsymbol{\Omega}_A \boldsymbol{d}_\rho^\top
           + \boldsymbol{d}_\eta    \boldsymbol{\Omega}_E   \boldsymbol{d}_\eta^\top  \\
    &\leq \Tr \rho H_A + \boldsymbol{d}_\rho \boldsymbol{\Omega}_A \boldsymbol{d}_\rho^\top
                       + \boldsymbol{d}_\eta \boldsymbol{\Omega}_E \boldsymbol{d}_\eta^\top  \\
    &\leq 2 P_{A} +  P_E . \label{eq:NstarHA}
\end{split}\end{equation}
Now, we choose $c$ such that $cc_E\leq 1$ and set
\begin{equation}
    H_B:=c \, \hat{\boldsymbol{r}}_B \hat{\boldsymbol{r}}^\top_B,
\end{equation}
which evidently is a positive self-adjoint operator independent of $\eta$ and $\rho$.
It trivially satisfies
\begin{equation}
      \Tr \mathcal{N}(\rho\ox\eta) H_B
        =\Tr \mathcal{N}_\eta(\rho) H_B
        =\Tr \rho   \mathcal{N}_\eta^*(H_B)
        =\Tr \rho   \mathcal{N}_\eta^*(H_B- c c_E  P_{E}\1)
            + \Tr \rho\mathcal{N}_\eta^*( c c_E   P_{E}\1 ),
\end{equation}
and thanks to Eq.~\eqref{eq:NstarHA}, we have
\begin{equation}
    \Tr \mathcal{N}(\rho\ox\eta) H_B
        \leq  \Tr \rho \mathcal{N}_\eta^*(H_{B, \eta}) + c c_E P_{E}
        \leq 2 P_{A}+ P_E + c c_E P_{E} \leq 2P_A+ 2P_E,
\end{equation}
concluding the proof.
\end{proof}


\section{Quantum communication}
\label{sec:quantum}

In this section we discuss the model of quantum communication with environment assistance.
We first focus on the unconstrained quantum capacity, for which we refer to isometries
giving rise to BGC for each choice of Gaussian initial environment state, and then move
one to energy-constrained quantum capacities.

Given a Gaussian isometry $W : AE \to BF$,
to send quantum information down the channel ${\cal N}_\eta(\rho)
=\tr_F W(\rho\otimes \eta)W^\dag$
from Alice to Bob, we need an encoding CPTP map ${\cal E}:{\cal T}(A_0)\to{\cal T}(A^n)$ and a
decoding CPTP map ${\cal D}:{\cal T}(B^n)\to{\cal T}(B_0)$,
where the number of qubits of $A_0$ is equal to that of $B_0$.
The output, upon inputting a maximally entangled state $\Phi^{RA_0}$ with
$R$ being an inaccessible reference system, reads
$\sigma^{RB_0} ={\cal D}\left({\cal N}^{\otimes n}\left({\cal E}\left(\Phi^{RA_0}\right)\otimes \eta^{E^n}\right)\right)$.

\begin{definition}\label{qcode}
A passive environment-assisted quantum code of block length $n$ is given by a triple
$\left( {\cal E}^{A_0\to A^n}, \eta^{E^n}, {\cal D}^{B^n\to B_0}\right)$ of an encoding
map, a helper state and a decoding map. Its fidelity is given by
$\tr \Phi^{RA_0}\sigma^{RB_0}$ and its rate by the number of qubits of $A_0$ divided by $n$.

A rate $\tt R$ is called achievable if there are codes for all block
lengths $n$ with fidelity converging to 1 and rate converging to $\tt R$.
The \emph{passive environment-assisted quantum capacity} of $W$, denoted $Q_H(W)$,
is the supremum of all achievable rates.

If the helper is restricted to fully separable states $\eta^{E^n}$,
i.e. convex combinations of tensor products
$\eta^{E^n} = \eta^{E^1}\otimes \ldots \otimes \eta^{E^n}$,
the supremum of all achievable rates is called
\emph{separable passive environment-assisted quantum capacity} and denoted $Q_{H\otimes}(W)$.

If in addition the helper is restricted to Gaussian states, we get
the \emph{Gaussian separable passive environment-assisted quantum capacity},
which we denote $Q_{GH\otimes}(W)$.
\end{definition}

\begin{theorem}
For a Gaussian isometry $W : AE \to BF$, the passive environment-assisted
quantum capacity is given by
\begin{equation}\label{eq:25}
\begin{split}
  Q_H(W) &= \sup_n\max_{\eta^{(n)} }\frac{1}{n} Q\left({\cal N}^{\ox n}_{\eta^{(n)}}\right)\\
         &= \sup_n\max_{\rho^{(n)},\eta^{(n)}} \frac{1}{n}
                                              I_c\left(\rho^{(n)}; {\cal N}^{\ox n}_{\eta^{(n)}}\right),
\end{split}
\end{equation}
where the maximization is over states $\rho^{(n)}$ on $A^n$ and states $\eta^{(n)}$ on $E^n$.

Similarly, the capacity with separable helper is given by the same formula,
\begin{equation}\label{eq:26}
\begin{split}
 Q_{H\ox}(W) &= \sup_n \max_{\eta_1\ox\ldots\eta_n}\frac{1}{n}
                                 Q\left({\cal N}_{\eta_1}\ox\cdots\ox {\cal N}_{\eta_n}\right)\\
             &= \sup_n\max_{\rho^{(n)},\eta^{(n)}}\frac{1}{n}
                                 I_c\left(\rho^{(n)}; {\cal N}^{\ox n}_{\eta^{(n)}}\right),
\end{split}
\end{equation}
but now varying only over product states
$\eta^{(n)} = \eta_1\ox\ldots\ox\eta_n$.
Consequently,
\begin{equation}
  Q_H(W) = \lim_{n\to \infty}\frac{1}{n}Q_{H\ox}(W^{\ox n}).
\end{equation}
\end{theorem}

\begin{proof}
It is known that the coherent information for nontrivial Gaussian channels without
constrained energy is finite \cite{B14}. However, relations \eqref{eq:25}
and \eqref{eq:26} without energy constraint may be infinite. To guarantee
their finiteness, one has to exploit energy constraints together with
subadditivity and concavity of von Neumann entropy.

The direct part, i.e. the "$\geq$" inequality, follows from the Lloyd-Shor-Devetak
theorem applied to the channel $\left(\mathcal{N}\right)_{\eta^{(n)}}$, to be
precise asymptotically many copies of this block channel, so that the i.i.d.
theorems apply (cf. \cite{W13}).

For the converse part, i.e. "$\leq$", the proof is like \cite{KSWY}, which is based on
the argument of \cite{BNS98,Sch96,SchNi96}. In other words, the coherent information
of a code of block length $n$ as input state together with helper uses of an arbitrary
state $\eta^{(n)}$ is smaller than the expression in \eqref{eq:25}.
\end{proof}


\subsection{Universal (anti-)degradability properties}
One of the main problems in  quantum information theory is to express the quantum
capacity by a single-letter formula. This can be done when the channel possesses
the (anti-)degradability property, which guarantees the additivity of the coherent
information \cite{DS05}. Here we want to understand, for a given two-mode Gaussian
unitary, whether or not this property can hold true irrespective of the environment state.

Recall that degradability of ${\cal N}^{A\to B}_{\eta_E}$ is defined by the
existence of a CPTP map $\Gamma^{B\to F}$ such that
\begin{equation}
  \label{degcond}
  \widetilde{\cal N}^{A\to F}_{\eta_E}= \Gamma^{B\to F} \circ {\cal N}^{A\to B}_{\eta_E}.
\end{equation}
Analogously, anti-degradability is defined by the existence of a map
$\bar{\Gamma}^{F\to B}$ such that
\begin{equation}
  \label{antidegcond}
  \bar{\Gamma}^{F\to B} \circ \widetilde{\cal N}^{A\to F}_{\eta_E} = {\cal N}^{A\to B}_{\eta_E}.
\end{equation}

\begin{remark}
\label{obs1}
By looking at the discussion in Subsection \ref{subsec:GU}, we can see
that any two-mode unitary $U^{(q)}$ with $q\geq 1/2$ is
degradable with respect to all Gaussian environment pure states;
we say that the unitary is \emph{Gaussian universally degradable}.
\end{remark}

This comes from the fact that for the Gaussian quantum channel
${\cal N}^{A\to B}_{\boldsymbol{ V}_E,q}$
we can find the required channel $\Gamma^{B\to F}$ in Eq.~\eqref{degcond} as
$\widetilde{{\cal N}}^{F\to B}_{\boldsymbol{ V}_E, \frac{2q-1}{q}}$, because
\begin{equation}
  \widetilde{{\cal N}}^{A\to F}_{\boldsymbol{ V}_E, q}
     = \widetilde{{\cal N}}^{B\to F}_{\boldsymbol{ V}_E, \frac{2q-1}{q}}\circ  {\cal N}^{A\to B}_{\boldsymbol{ V}_E, q}.
\end{equation}

\begin{remark}
\label{obs2}
By looking at the discussion in Subsection \ref{subsec:GU}, we can see
that any two-mode unitary $U^{(q)}$ with $0\leq q \leq 1/2$ is anti-degradable with
respect to all Gaussian environment pure states;
we say that the unitary is \emph{Gaussian universally anti-degradable}.
\end{remark}

This comes from the fact that for the Gaussian quantum channel
${\cal N}^{A\to B}_{\boldsymbol{ V}_E,q}$
we can find the required channel $\bar{\Gamma}^{F\to B}$ in Eq.~\eqref{antidegcond} as
$\widetilde{{\cal N}}^{F\to B}_{\boldsymbol{ V}_E, \frac{1-2q}{1-q}}$, because
\begin{equation}
  {\cal N}^{A\to B}_{\boldsymbol{ V}_E, q}
   =\widetilde{{\cal N}}^{F\to B}_{\boldsymbol{V}_E,\frac{1-2q}{1-q}}\circ \widetilde{\cal N}^{A\to F}_{\boldsymbol{V}_E,q}.
\end{equation}

\begin{definition}
A two-mode Gaussian unitary $U$ is said to be \emph{universally degradable}
(resp. \emph{universally anti-degradable}) if
Eq.~\eqref{degcond} (resp. \eqref{antidegcond})
holds true for all environment states $\eta_E$.
\end{definition}

\begin{theorem}
\label{thm:nongaussian}
Any two-mode Gaussian unitary $U^{(q)}$
is neither universally degradable, nor universally anti-degradable, unless $q=1$.
\end{theorem}

The proof of this theorem, which we give in Appendix \ref{app:thmproof},
is obtained by assuming the existence of a quantum channel $\Gamma$ satisfying the
degradability condition \eqref{degcond} and then showing that this leads to a contradiction.
In particular, for $q\leq 1/2$ the claim follows from the fact that the channel is
Gaussian universally anti-degradable, but has positive coherent information,
and hence cannot be anti-degradable,
for some non-Gaussian environment states \cite{LPGH}.

\begin{corollary}
\label{th3}
The two-mode Gaussian unitary $U^{(q)} : AE \to BF$ with $q\geq 1/2$
is Gaussian universally degradable, and hence its
Gaussian separable passive environment-assisted quantum capacity
is given by the single-letter formula
\begin{equation}
  Q_{GH\ox}(U^{(q)}) = \max_{\eta_G}\sup_{\rho} I_c(\rho_G;{\cal N}_\eta),
\end{equation}
where the optimization can be restricted to Gaussian input states $\rho_G$ (cf.~\cite[Thm.~12.38]{H12}).
Note that for each fixed $\eta_G$, the coherent information $I_c$ is
a concave function of the covariance matrix of $\rho_G$,
thus it is sufficient to find a local maximum which necessarily must be the global one.

For $q \leq 1/2$, the two-mode Gaussian unitary $U^{(q)} : AE \to BF$
is Gaussian universally anti-degradable, and hence its
Gaussian separable passive environment-assisted quantum capacity vanishes,
$Q_{GH\ox}(U^{(q)}) = 0$.
\hfill$\square$
\end{corollary}

\medskip
Armed with this corollary, we can now proceed to calculate the
Gaussian separable passive environment-assisted quantum capacity of
the two-mode unitaries $U^{(q)}$. Note that for each Gaussian environment
state $\eta_G$, the resulting channel $\mathcal{N}_\eta$ is an OMG,
a one-mode Gaussian channel. Their complete classification is given in \cite{H07}.
In particular, when $\eta=\proj{0}$ is the vacuum state,
$U^{(q)}$ gives rise to an attenuator channel for $q<1$,
and an amplifier channel for $q>1$; for $q=1$, $\mathcal{N}_{\proj{0}}$ is
the identity.

For an OMG channel described by Eq. \eqref{dVtransforms},
the parameters that characterize it are
\begin{equation}
  \label{xandy}
  x:=\sqrt{\det \boldsymbol{X}}, \quad y:= \det \boldsymbol{Y}.
\end{equation}
Furthermore, we define another parameter dependent on these two,
$K := \frac12 \left(y-|1-x|\right)$.

For OMG channels, whenever the coherent information is non-zero, the
supremum over all Gaussian input states is achieved for infinite input
power, $P_A \to\infty$.
It is known from \cite{B14} that the optimised coherent information
(over all Gaussian input states) is given by
\begin{equation}
  \label{eq:Brad}
  \sup_{\rho_G}I_c(\rho_G;{\cal N})
        = \frac{K}{|1-x|}\ln\frac{K}{|1-x|}
            - \frac{K+|1-x|}{|1-x|}\ln\frac{K+|1-x|}{|1-x|}
            + \ln\frac{x}{|1-x|}.
\end{equation}
For $0\leq q\leq 1$, $U^{(q)}$ with the symplectic matrix \eqref{Scan1}
describes a beam splitter with transmissivity $q$.
Considering $1/2\leq q <1$, then from Corollary \ref{th3} we have
\begin{equation}
  Q_{GH\ox}(B(q)) = \max_{{\boldsymbol{ V}}_E} \sup_{{\boldsymbol{ V}}_A} I_c\left({\boldsymbol{ V}}_A; B(q)\right),
\end{equation}
where the maximization over environment states can be restricted to pure one-mode
states given by the covariance matrix
\begin{equation}\label{eq:VEsq}
{\boldsymbol{V}}_E=\begin{pmatrix}
\cosh(2s)+\cos\theta\cosh(2s) & \sin\theta\sinh(2s) \\
\sin\theta\sinh(2s) & \cosh(2s)-\cos\theta\cosh(2s)
\end{pmatrix},
\end{equation}
with $s\in\mathbb{R}$ and $\theta\in[0,2\pi)$.
Eqs. \eqref{Scan1} and \eqref{xandy} yield $x =q$ and $y =1-q$ for all one-mode
squeezed input environment ${\boldsymbol{V}}_E$. Invoking Eq.~\eqref{eq:Brad}, we get
\begin{equation}
  Q_{GH\ox}(B(q))=\ln\frac{q}{1-q}.
\end{equation}
For $q>1$, $U^{(q)}$ is a two-mode squeezing transformation with gain $q$,
which has the symplectic matrix \eqref{Scan2}.
Then from Corollary \ref{th3} we have
\begin{equation}
Q_{H\ox}(A(a))=\max_{{\boldsymbol{V}}_E}\sup_{{\boldsymbol{ V}}_A}
I_c\left( {\boldsymbol{ V}}_A; A(q) \right),
\end{equation}
where the maximization over environment states can again be restricted to states of
the form \eqref{eq:VEsq}.
Eqs. \eqref{Scan2} and \eqref{xandy} yield $x =q$ and $y =q-1$ for all one-mode
squeezed input environment ${\boldsymbol{V}}_E$. Invoking \eqref{eq:Brad}, we get
\begin{equation}
  Q_{GH\ox}(A(q))=\ln\frac{q}{q-1}.
\end{equation}

Both for $q>1$ and $q<1$, the formulas recover the infinite capacity
of the identity channel in the limit $q\rightarrow 1$.


\subsection{Energy-constrained passive-environment assisted quantum capacities}

We now move on to energy-constrained quantum capacities.
Suppose that $P_A$ (resp. $P_E$) is the maximum allowed average energy per mode
on $A$ system (resp. $E$ system). Then we modify the Definition \ref{qcode} as follows.

\begin{definition}
An energy constrained passive environment assisted quantum code of block length $n$ is a triple
$\left( {\cal E}^{A_0\to A^n}, \eta^{E^n}, {\cal D}^{B^n\to B_0}\right)$
such that, $\tr\left[\tr_R{\cal E}
\left(\Phi^{RA_0}\right)\right]H_{A^n}\leq nP_A$ and  $\tr \eta^{(n)}H_{E^n} \leq n P_E$.

Its fidelity is given by
$\tr \Phi^{RA_0}\sigma^{RB_0}$ and its rate by the number of modes of $A_0$ over $n$.

A rate $\tt R$ is called achievable if there are codes for all block
lengths $n$ with fidelity converging to 1 and rate converging
to $\tt R$. The energy constrained passive environment-assisted quantum capacity
of $W$, denoted $Q_H(W;P_A;P_E)$ is the supremum
of all achievable rates.

If the helper is restricted to fully separable states
$\eta^{E^n}$,
i.e. convex combinations of tensor products
$\eta^{E^n} = \eta^{E^1}\otimes \ldots \otimes \eta^{E^n}$,
the supremum of all achievable rates is denoted
$Q_{H\otimes}(W;P_A;P_E)$.
\end{definition}


\begin{theorem}\label{thm:constrQ}
For a Gaussian isometry $W : AE \to BF$,
the energy-constrained passive environment-assisted quantum capacity
is given by
\begin{equation}
\begin{split}
  Q_H(W; P_A; P_E)
      &=\sup_n\sup_{\eta^{(n)}} \frac{1}{n} Q\left({\cal N}_{\eta^{(n)}}^{\ox n},nP_A\right) \\
      &=\sup_n  \sup_{\eta^{(n)}} \max_{\rho^{(n)}} \frac{1}{n} I_c\left(\rho^{(n)};
                                                               {\cal N}_{\eta^{(n)}}^{\ox n}\right),
\end{split}
\end{equation}
where the maximization is over states $\rho^{(n)}$ on $A^n$ with
$\tr \rho^{(n)}H_{A^n}\leq n P_A$ and states
$\eta^{(n)}$ on $E^n$ with  $\tr \eta^{(n)}H_{E^n} \leq n P_E$.

The capacity with separable helper is given by the same
formula, but now varying only over product states $\eta^{(n)}=\eta_1\ox\ldots\ox\eta_n$
and respecting the energy constraints
$\tr \rho^{(n)}H_{A^n} \leq n P_A$ and $\sum_{i=1}^n \tr \eta_i H_{E_i} \leq n P_E$.
Consequently,
$\displaystyle{Q_{H}(W; P_A; P_E) = \lim_{n\to\infty}\frac{1}{n}Q_{H\ox}(W; nP_A; nP_E)}$.
\end{theorem}

\begin{proof}
Considering the Hamiltonian operator $H_{A^nE^n}=H_{A^n}\ox\1_{E^n}+ \1_{A^n}\ox H_{E^n}$
on the system $A^nE^n$, we have
\begin{equation}
  \Tr \rho^{(n)}\ox\eta^{(n)} H_{A^nE^n} \leq nP_A + nP_E,
\end{equation}
where $\rho^{(n)}\ox\eta^{(n)}$ is an arbitrary allowed input state to the system $A^nE^n$.
Using the fact that
\begin{equation}
 \Tr\exp(-\beta H_{A^n}),\ \Tr\exp(-\beta H_{E^n}) < \infty \text{ for all } \beta>0,
\end{equation}
we get
\begin{equation}
  \Tr\exp(-\beta H_{A^nE^n})
       = \bigl(\Tr\exp(-\beta H_{A^n})\bigr) \bigl(\Tr\exp(-\beta H_{E^n})\bigr)
       < \infty.
\end{equation}
Thus, according to \cite{H12}, the set
$\mathcal{C}=\{ \rho^{(n)}\ox\eta^{(n)} : \Tr(\rho^{(n)}\ox\eta^{(n)} H_{A^nE^n})\leq nP_A + nP_E  \}$
is compact.

Using \cite[Cor.~14]{sh2015} and the fact that
\begin{equation}\label{eq:energyco}
 \sup_{\rho^{(n)}\ox\eta^{(n)} : \Tr (\rho^{(n)}\ox\eta^{(n)}) H_{A^nE^n} \leq   nP_A + nP_E}
                                          S(\mathcal{N}^{\otimes n}(\rho^{(n)}\ox\eta^{(n)}))<\infty,
\end{equation}
coming from Lemma \ref{lem:lemma2}, we see that the coherent information
$I_c\left(\rho^{(n)}; {\cal N}_{\eta^{(n)}}^{\ox n}\right)$, for any fixed $\eta^{(n)}$, is continuous and hence it takes its maximum on the set $\{\rho^{(n)}\,\vert \,  \tr \rho^{(n)}H_{A^n}\leq n P_A\}$.  By applying \eqref{eq:energyco}, we then have
$$-\infty < -S({\cal N}_{\eta^{(n)}}^{\ox n}(\rho^{(n)}))\leq I_c\left(\rho^{(n)}; {\cal N}_{\eta^{(n)}}^{\ox n}\right)\leq S({\cal N}_{\eta^{(n)}}^{\ox n}(\rho^{(n)}))<+\infty.$$
Therefore, the quantity $Q_H(W; P_A; P_E)$ is finite.
\end{proof}

\medskip
\begin{remark}
If $\eta^{(n)}$ is pure, we have $I_c\left(\rho^{(n)}; {\cal N}_{\eta^{(n)}}^{\ox n}\right)
 =I_c\left(\rho^{(n)}\ox \eta^{(n)}; {\cal N}^{\ox n}\right)$ and the latter is
 is continuous with maximum on $\mathcal{C}$.
 As consequence the  $\sup_{\eta^{(n)}}$ in Theorem \ref{thm:constrQ} can be turned into
  $\max_{\eta^{(n)}}$.
\end{remark}

\medskip
Let us evaluate the energy-constrained environment-assisted quantum capacities
for unitaries that are universally degradable with respect to Gaussian environment states.
To do so, recall \cite[Thms.~13 and 14]{WQ16} that for a degradable channel
${\cal N}^{A\to B}$, the energy-constrained quantum capacity is given by
\begin{equation}
  Q({\cal N},P_A) = \sup_{\rho\,:\,\tr \rho H_A\leq P_A}
                          S\bigl({\cal N}(\rho)\bigr) -S\bigl(\widetilde{\cal N}(\rho)\bigr),
\end{equation}
where the supremum is achieved by the Gibbs state $\gamma_A(P_A)$.

In particular for degradable channels ${\cal N}_i$,
\begin{equation}
  Q({\cal N}_1\ox\ldots\ox {\cal N}_n, nP)
    = \max_{\{P_i\}}
        \sum_i S\left({\cal N}_i(\gamma_A(P_i))\right)-S\left(\widetilde{\cal N}_i(\gamma_A(P_i))\right)
          \text{ s.t. } \sum_i P_i = n P_A ,
\end{equation}
an optimization that can be performed by Lagrange multipliers in the cases of interest.

For unitaries that are universally degradable with respect to Gaussian environment states,
the energy-constrained Gaussian separable environment-assisted capacity is bounded below by
\begin{equation}
  Q_{GH\ox}(U,P_A,P_E)
    \geq \max_{\eta_G \,:\, \tr\eta H_E \leq P_E}
                      S\left(\tr_F U_{\eta_G}(\gamma_A(P_A))\right)
                      -S\left( \tr_B U_{\eta_G}(\gamma_A(P_A))\right).
\end{equation}
With this we can find lower bounds for beam splitter and amplifier unitaries,
and additionally also find their upper bounds when letting $P_A\to\infty$.


\section{Classical communication}
\label{sec:classical}

In this section we consider classical communication in the passive
environment-assisted model.
After deriving the classical capacity, we put forward an uncertainty relation
for it, that arises when exchanging the roles of active and passive users.
Finally we will briefly discuss conferencing encoders.

Suppose Alice selects some classical message $m$ from the set of messages
$\{1, 2, \ldots, |M|\}$ to communicate to Bob.
An encoding CPTP map ${\cal E} : M \to {\cal T}(A^n)$ can be realized by preparing states
$\{\alpha_m\}$ to be input across $A^n$ of $n$ instances of the channel.
Here $M$ is an Hilbert space with orthonormal basis $\{|m\rangle\}$.
A decoding CPTP map ${\cal D}: {\cal T}(B^n)\to M$ can be
realized by a positive operator-valued measure (POVM) $\{\Lambda_m\}$.
The probability of error for a particular message $m$ is
\begin{equation}
P_e(m) = 1- \Tr\left[\Lambda_m {\cal N}^{\otimes n}\left(\alpha_m^{A^n}
\otimes \eta^{E^n}\right)\right].
\end{equation}

\begin{definition} A passive environment-assisted classical code of block length n is a family of
triples $\{\alpha_M^{A^n}, \eta^{E^n}, \lambda_m\}$ with the error probability $\overline{P}_e :=
\frac{1}{|M|}\sum_m P_e(m)$ and the rate $\frac{1}{n}\ln |M|$.
A rate $\tt R$ is achievable if there is a sequence of codes over their block length $n$ with
$\overline{P}_e$ converging to 0 and rate converging to $\tt R$. The passive environment-assisted classical capacity of $W$, denoted by $C_H(W)$, is the maximum achievable rate.

If the helper is restricted to fully separable states $\eta^{E^n}$, i.e., convex combinations of tensor products $\eta^{E^n}= \eta^{E_1}\ox \ldots \eta^{E_n}$, the
largest achievable rate is denoted by $C_{H\ox}(W)$.
\end{definition}

Since the error probability is linear in the environment state, without loss of generality the latter may
be assumed to be pure, for both unrestricted and separable helper.

\begin{theorem}
\label{thm:cc}
For a Gaussian isometry $W : AE \to BF$,
the energy-constrained passive environment-assisted classical capacity is given by
\begin{equation}
  C_H(W,P_A,P_E)=\sup_n\max_{\eta^{(n)}}\frac{1}{n}C\left({\cal N}^{\ox n}_{\eta^{(n)}},nP_A\right),
\end{equation}
where the maximization is over environment input states $\eta^{(n)}$ respecting
energy constraint $\tr \eta^{(n)} H_{E^n} \leq nP_E$.

Similarly, the capacity with separable helper is given by the same formula,
\begin{equation}
  C_{H\ox}(W,P_A,P_E)
     =\sup_n\max_{\eta^{(n)}=\eta_1\ox\ldots\ox\eta_n}
                  \frac{1}{n}C\left({\cal N}_{\eta_1}\ox\ldots\ox{\cal N}_{\eta_n},nP_A\right),
\end{equation}
where the maximum is only over product states,
i.e. $\eta^{(n)}=\eta_1\ox\ldots\ox\eta_n$
respecting the energy constraint $\tr \eta^{(n)}H_{E^n} \leq nP_E$.
\end{theorem}

As a consequence of the theorem, we have
$C_H(W,P_A,P_E)=\lim_{n\to\infty}\frac{1}{n}C_{H\ox}(W,nP_A,nP_E)$.

\medskip
\begin{proof}
Consider the Hamiltonian operator $H_{AE}=H_{A}\otimes \1 + \1\otimes H_{E}$ on the system $AE$ together with
\begin{equation}
H_{A^nE^n} := H_{AE}\otimes \cdots\otimes \1
               + \1 \otimes H_{AE}\otimes \cdots\otimes \1
               + \ldots
               + \1\otimes \cdots \otimes H_{AE}.
\end{equation}
Consider further $\eta_{\boldsymbol{i}}^{(n)}=\eta_{i_1}\otimes \eta_{i_2}\otimes \cdots \otimes \eta_{i_n}$, where $\boldsymbol{i}$ is a cyclic permutation. Then we have
\begin{align}
\sup_{\rho^{(n)}: \Tr \rho^{(n)} H_{A^n} \leq n P_A} S\left(  {\cal N}^{\ox n}_{\eta^{(n)}}\left(\rho^{(n)}\right) \right)
&\leq
  \sup_{\rho^{(n)}: \Tr \rho^{(n)} H_{A^n} \leq n P_A}\sum_{\boldsymbol{i} \, cyclic} S\left(  {\cal N}^{\ox n}_{\eta_{\boldsymbol{i}}^{(n)}}\left(\rho^{(n)}\right) \right)\label{eq:87-1}\\
  &\leq\sup_{\rho^{(n)}: \Tr \rho^{(n)} H_{A^n} \leq n P_A} \sum_{\boldsymbol{i} \, cyclic} S\left(  {\cal N}^{\ox n}\left(\rho^{(n)}\otimes {\eta_{\boldsymbol{i}}^{(n)}}\right) \right)\\
  &\leq \sup_{\bar{\rho}: \Tr \bar{\rho} H_{A} \leq  P_A} \sum_{i, j=1}^n  S\left(  {\cal N}\left(\rho_j\otimes \eta_{i, j}\right) \right)\\
   &\leq n\sup_{\rho^{(n)}: \Tr \rho^{(n)} H_{A^n} \leq n P_A}\sum_{i=1}^n   S\left(  {\cal N}\left(\bar{\rho}\otimes \eta_{i}\right) \right)  \label{eq:85} \\
   &\leq n^2 \sup_{\bar{\rho}: \Tr \bar{\rho} H_{A} \leq  P_A }  S\left(  \cal N\left(\bar{\rho}\otimes \bar{\eta}\right) \right)
\end{align}
where $\bar{\rho}=\frac{1}{n}\sum_{i=1}^{n} \rho_i$ and $\bar{\eta}=\frac{1}{n}\sum_{i=1}^{n} \eta_i$.
In getting the above sequence of inequality we exploited the subadditivity and the concavity of the von Neumann entropy,

Next, we have
\begin{align}
\sup_{\eta^{(n)}: \Tr (\eta^{(n)} H_{E^n}) \leq n P_E}\sup_{\rho^{(n)}: \Tr (\rho^{(n)} H_{A^n}) \leq n P_A} S\left(  {\cal N}^{\ox n}_{\eta^{(n)}}\left(\rho^{(n)}\right) \right)
&\leq
n^2 \sup_{\bar{\eta}: \Tr \bar{\eta} H_{E} \leq  P_E}\sup_{\bar{\rho}: \Tr \bar{\rho} H_{A} \leq  P_A }  S\left(  \cal N\left(\bar{\rho}\otimes \bar{\eta}\right)\right) \label{eq:87}
\end{align}
By Lemma \ref{lem:lemma2} the quantity \eqref{eq:87} is finite and so is the l.h.s. of \eqref{eq:87-1}.

Let us consider $\rho=\int p_x\rho_x dx$ as the average input on a single channel use. Clearly we have
\begin{equation}
  \Tr \rho H_A\leq \int \Tr \rho_x H_A p_x dx \leq P_{A}.
\end{equation}
Replacing $\rho$ by $\mathcal{N}_\eta(\rho)$ and using Lemma \ref{lem:lemma2}, the Holevo
$\chi$-quantity
\begin{equation}
   \chi\left( \left\{ p_x, {\cal N}_\eta\left( \rho_x \right)\right\}\right)=S\left(\mathcal{N}_\eta(\rho)\right)
   -\int S(\mathcal{N}_\eta(\rho_x)) p_x dx,
\end{equation}
results finite. Then by means of \eqref{eq:87}, it is clear that
\begin{equation}
C\left({\cal N}^{\ox n}_{\eta^{(n)}},nP_A\right)=\sup_{p_{x^n}, \,\rho_{x^{n}}^{(n)}}
\chi\left( \left\{ p_{x^n}, {\cal N}_\eta^{\otimes n}\left( \rho_{x^{n}}^{(n)} \right)\right\}\right),
\end{equation}
is finite as well and so $C_H$ is correctly defined. Now,
 the proof of the direct parts, i.e. ``$\geq$'',
follows immediately from the Holevo-Schumacher-Westmoreland theorem \cite{H98,SchW97}.

For the converse parts, i.e. ``$\leq$'',
the proof goes like \cite[Thm.~1]{KSWY2}.
\end{proof}

\medskip
For unitaries of most interest, like beam-splitter and amplifier, we can
give a lower bound on the classical capacity with separable helper.
Let us encode classical stochastic variable $m$, distributed according
to a probability density $P_m$, into the quantum states $\rho_m^A$.
The modulation due to encoding is given by $\boldsymbol{V}_{mod}$ and
$\overline{\boldsymbol{V}}_A = \boldsymbol{V}_A + \boldsymbol{V}_{mod}$
gives the average input state after encoding.
We assume that the distribution of the classical messages is a Gaussian
distribution with zero mean whose covariance matrix is given by
$\boldsymbol{V}_{mod}$.
The average energy of the input states in terms of the CM is
given by $P_A = \frac{\tr \overline{\boldsymbol{V}}_A}{4n}-\frac{1}{2}$ ,
and likewise $P_E = \frac{\tr \overline{\boldsymbol{ V}}_E}{4n}-\frac{1}{2}$
for the environment.
Then, for beam splitter and amplifier we get the following form for the
environment-assisted capacities when the helper is restricted to separable
states in the environment,
\begin{equation}
C_{H\ox}(U,P_A,P_E)\geq \max_s \left\{
g\left(|x|P_A+y\cosh(2s)+\frac{|x|-1}{2}\right)
-g\left(y+\frac{|x|-1}{2}\right);
P_A\geq P_{\rm th}\right\},
\end{equation}
where we used the notations $x$ and $y$ from Eq. \eqref{xandy}
with $x \neq 0, 1$. Furthermore, $\cosh(2s)\leq 2P_E+1$ and
$P_{\rm th}=e^{2|s|}+\frac{2y\sinh(2|s|)}{|x|}-1$.
For a general one-mode environment state we can find a symplectic orthogonal
transformation, that makes $\boldsymbol{V}_E$ diagonal (this symplectic
orthogonal transformation is a rotation, thus the effective state is
a squeezed one-mode state), which does not affect the energy constraints
on the input environment. Now using \cite[Thm.~1]{PLM12}, we have
$\boldsymbol{V}_A$ and $\boldsymbol{V}_{mod}$ to be diagonal in the same basis as
$\boldsymbol{V}_E$. In fact we can choose the seed state of the input to be
$\boldsymbol{V}_E$ (in its diagonal form). Then following the calculation
in \cite{SKG13}, we get the claimed result.


\subsection{Capacities uncertainty relation}

For a given isometry $W:AE\to BF$, the following quantity corresponds to the
\emph{product-state capacity with separable helper}
\begin{equation}
  \chi_{H\ox}(W,P_A,P_E)
     =\max_{\rho,\eta\,:\, \tr \rho H_A\leq P_A,\, \tr \eta H_E \leq P_E}
                                    \chi\left( \left\{ p_x dx, {\cal N}_\eta\left( \rho_x \right)\right\}\right),
\end{equation}
where on the r.h.s we have the
Holevo $\chi$ quantity for the effective channel
${\cal N}^{A\to B}_\eta (\rho) := {\cal N}^{AE\to B}(\rho \otimes \eta)$
[see Eq.~\eqref{eq:Neta}] upon inputting the ensemble $\{p_x dx, \rho_x\}$,
and $\rho=\int p_x\rho_x dx$.

Now, besides this channel ${A\to B}$,
we can also define another effective channel
${E\to B}$ by fixing the state of $A$ and tracing over $F$, namely
$\overline{{\cal N}}^{E\to B}_\rho (\eta) := {\cal N}^{AE\to B}(\rho \otimes \eta)$
[see again Eq.~\eqref{eq:Neta}].
For this latter the following quantity corresponds to the
product-state capacity with separable helper
\begin{equation}
\chi_{A\ox}(W,P_A,P_E)
  =\max_{\rho,\eta \,:\, \tr \rho H_A \leq P_A,\, \tr \eta H_E \leq P_E}
   \chi\left( \left\{ p_x dx, \overline{{\cal N}}_\rho\left( \eta_x \right)\right\}\right).
\end{equation}

\begin{theorem}
\label{thm:uncertainty}
Given a Gaussian unitary $W :A \ox E \longrightarrow B \ox F$,
together with
${\rm rank}(\boldsymbol{N}\boldsymbol{\Sigma}_{2N}\boldsymbol{N}^\top)=2N$,
assuming Hamiltonians $H_A$ and $H_E$ for Alice and the helper as in Eq.~\eqref{eq:HX},
with an average photon number per mode constrained by $P_A$ and $P_E$ respectively, we have
\begin{equation}
\label{eq:uncertainty}
  \chi_{A\ox}(W,P_A,P_E) + \chi_{H\ox}(W,P_A,P_E) \geq
  \frac{ \min\{P_A, P_E\}}{2\max\{P_E, P_A\}+1}.
\end{equation}
\end{theorem}

\begin{remark}
This is a kind of uncertainty relation for $\chi_{H\ox}$ and $\chi_{A\ox}$,
reminiscent of the entropic uncertainty relations for complementary
observables (see e.g. \cite{WW10})
saying that not both of them can be arbitrary small.
\end{remark}

\medskip

\begin{proof}
Since the involved capacities refer to product states with separable helper,
we can consider single systems $A$, consisting of $N_A$ modes, and $E$,
consisting of $N_E$ modes.

From the relation \eqref{eq:22}, the covariance matrix of input state for
system $A$ changes to the following
\begin{equation}\label{eq:VAVB}
  \boldsymbol{V}_A \mapsto \boldsymbol{V}_B
                            = \boldsymbol{M} \boldsymbol{V}_A \boldsymbol{M}^\top
                              + \boldsymbol{N} \boldsymbol{V}_E \boldsymbol{N}^\top.
\end{equation}
Instead, considering as input the system $E$ and as helper $A$,
the corresponding output is obtained by exchanging $A$ and $E$
in the above expression, namely
\begin{equation}\label{eq:VEVB}
  \boldsymbol{V}_E \mapsto \boldsymbol{V}_B
                           = \boldsymbol{M} \boldsymbol{V}_E \boldsymbol{M}^\top
                             + \boldsymbol{N} \boldsymbol{V}_A \boldsymbol{N}^\top.
\end{equation}
As input ensembles, we consider coherent states subject to
Gaussian distributions with zero mean, whose covariance matrix are given by
$\boldsymbol{V}_{A,mod}$ and $\boldsymbol{V}_{E,mod}$ for
Alice and Helen, respectively. The action of encoding is described as follows:
\begin{equation}
\begin{split}
  \overline{\boldsymbol{V}}_A &= \boldsymbol{V}_A+\boldsymbol{V}_{A,mod},\\
  \overline{\boldsymbol{V}}_E &= \boldsymbol{V}_E+\boldsymbol{V}_{E,mod}.
\end{split}
\end{equation}
The respective average output states are then given by
\begin{align}
  \label{eq:aveAB}
  \overline{\boldsymbol{V}}_A \mapsto \overline{\boldsymbol{V}}_B
                                &= \boldsymbol{M}\overline{\boldsymbol{ V}}_A\boldsymbol{M}^\top
                                    + \boldsymbol{N}\overline{\boldsymbol{V}}_E\boldsymbol{N}^\top,\\
  \label{eq:aveEB}
  \overline{\boldsymbol{ V}}_E \mapsto \overline{\boldsymbol{V}}_B
                                &= \boldsymbol{M}\overline{\boldsymbol{ V}}_E\boldsymbol{M}^\top
                                   + \boldsymbol{N}\overline{\boldsymbol{V}}_A\boldsymbol{N}^\top.
\end{align}
Coherent input states on the systems $A$ and $E$ means
$\boldsymbol{V}_A = \boldsymbol{V}_E = \frac{1}{2} \boldsymbol{I}$.
Then, using Eqs. \eqref{eq:VAVB} and \eqref{eq:aveAB}, we can find
\begin{equation}
\chi_{H\ox}
\geq
S\left( \boldsymbol{M} \left(\frac{1}{2}\boldsymbol{I}
+ {\boldsymbol{ V}}_{A,mod}  \right) \boldsymbol{M}^\top
+\frac{1}{2}\boldsymbol{N}\boldsymbol{N}^\top \right)
- S\left(\frac{1}{2}\boldsymbol{M} \boldsymbol{M}^\top
+\frac{1}{2} \boldsymbol{N}\boldsymbol{N}^\top \right) ,
\end{equation}
and analogously using Eqs. \eqref{eq:VEVB} and \eqref{eq:aveEB}, we can find
\begin{equation}
\chi_{A\ox}
\geq S\left(\boldsymbol{N} \left(\frac{1}{2}\boldsymbol{I}
+ {\boldsymbol{ V}}_{E,mod}  \right) \boldsymbol{N}^\top
+\frac{1}{2}\boldsymbol{M} \boldsymbol{ M}^\top\right)
- S\left(\frac{1}{2}\boldsymbol{M}  \boldsymbol{M}^\top
+\frac{1}{2} \boldsymbol{N}  \boldsymbol{N}^\top\right).
\end{equation}
Choosing ${\boldsymbol{ V}}_{A,mod}= P_A \boldsymbol{I}$
for the channel $\mathcal{N}$ and  ${\boldsymbol{ V}}_{E,mod}= P_E \boldsymbol{I}$ for the channel $\mathcal{M}$, we get
\begin{equation}\label{eq:lbchiN}
\chi_{H\ox}
\geq S\left(   \left(P_A+\frac{1}{2}\right) \boldsymbol{M} \boldsymbol{M}^\top
+\frac{1}{2} \boldsymbol{N}  \boldsymbol{N}^\top\right)
- S\left(\frac{1}{2}\boldsymbol{M}  \boldsymbol{M}^\top
+\frac{1}{2} \boldsymbol{N}  \boldsymbol{N}^\top \right) ,
\end{equation}
and
\begin{equation}\label{eq:lbchiM}
\chi_{A\ox}
\geq S\left( \left(P_E+\frac{1}{2}\right) \boldsymbol{N} \boldsymbol{N}^\top
+\frac{1}{2}\boldsymbol{M}  \boldsymbol{M}^\top\right)
- S\left(\frac{1}{2}\boldsymbol{M}  \boldsymbol{M}^\top
+\frac{1}{2} \boldsymbol{N}  \boldsymbol{N}^\top \right).
\end{equation}

Now define the functions
\begin{equation}
f(t) := \str\left( t \boldsymbol{M} \boldsymbol{M}^\top
+ \frac{1}{2} \boldsymbol{N} \boldsymbol{N}^\top\right),
\quad\text{ and }\quad
h(t) := \str\left( \frac{1}{2} \boldsymbol{M} \boldsymbol{M}^\top
+ t \boldsymbol{N} \boldsymbol{N}^\top\right) = 2t f\left(\frac{1}{4t}\right),
\end{equation}
where $\str$ denotes the symplectic trace, i.e. $\str(A)=\sum_i \nu_i(\boldsymbol{A})$,
with $\nu_i(\boldsymbol{A})$ the symplectic eigenvalues of $\boldsymbol{A}$.
Notice that these functions are strictly increasing with respect to the
parameter $t$, and so they are invertible functions.

By the Cauchy-Lagrange mean value theorem, for the function $g(x)$, we know there exists a
$c\in(a, b)$ such that
\begin{equation}
  \label{eq:Lag}
  g(b)-g(a)=g'(c)(b-a).
\end{equation}
Thus by choosing $t_b=f^{-1}(b)$, $t_a=f^{-1}(a)$ and $c_1=f^{-1}(c)$, we get
\begin{equation}
  g(f(t_b))-g(f(t_a))=g'(f(c_1)) (f(t_b)-f(t_a)).
\end{equation}
Consequently we can write
\begin{align}
 g\left(f\left(P_A+\frac{1}{2}\right)\right)-g\left(f\left(\frac{1}{2}\right) \right)
 &= \left[ f\left(P_A+\frac{1}{2}\right)- f\left( \frac{1}{2}\right) \right]
             \ln\left( \frac{f(c_1)+\frac{1}{2}}{f(c_1)-\frac{1}{2}} \right) \label{eq:14}\\
 &\geq \frac{f\left(P_A+\frac{1}{2}\right)- f\left( \frac{1}{2}\right) }{f(c_1)} \label{eq:15}\\
 &\geq \frac{f\left(P_A+\frac{1}{2}\right)- f\left( \frac{1}{2}\right) }{f\left(P_A+\frac{1}{2}\right)} \label{eq:16}\\
 &\geq \frac{1}{P_A+\frac{1}{2}}\frac{ f\left(P_A+\frac{1}{2}\right)- f\left( \frac{1}{2}\right)}{f\left(\frac{1}{2}\right)}.\label{eq:17}
\end{align}
From Eqs.~\eqref{eq:14} to \eqref{eq:15} we used the elementary relation
$x \ln \frac{x+\frac{1}{2}}{x-\frac{1}{2}}\geq 1$, valid for $x\geq \frac{1}{2}$.
From Eqs.~\eqref{eq:15} to \eqref{eq:17}  we used the property $\str(\boldsymbol{A})\geq \str(\boldsymbol{B})$, valid for symplectic matrices $\boldsymbol{A}$ and $\boldsymbol{B}$ such that  $\boldsymbol{A}\geq\boldsymbol{B}$ \cite{BAT}.
Analogously, by choosing in Eq.~\eqref{eq:Lag},
$t_b=h^{-1}(b)$, $t_a=h^{-1}(a)$ and $c_2=h^{-1}(c)$, we get
\begin{equation}
  g(h(t_b))-g(h(t_a))=g'(h(c_2)) (h(t_b)-h(t_a)).
\end{equation}
As a consequence, we can write
\begin{equation}\label{eq:gh}
 g\left(h\left(P_E+\frac{1}{2}\right)\right)-g\left(h\left(\frac{1}{2}\right) \right)
\geq \frac{1}{P_E +\frac{1}{2}}\frac{ h\left(P_E+\frac{1}{2}\right)- h\left( \frac{1}{2}\right) }{h\left(\frac{1}{2}\right)}.
\end{equation}

Assuming for the moment $P_A\leq P_E$, and taking into account that
$f$ and $h$ are increasing functions, together with the fact that
$f\left( \frac{1}{2}\right) =h\left( \frac{1}{2}\right)$, we obtain from
Eqs. \eqref{eq:17} and \eqref{eq:gh}
\begin{align}
  \label{eq:gfgh}
  g\left(f\left(P_A+\frac{1}{2}\right)\right)&-g\left(f\left(\frac{1}{2}\right) \right) +g\left(h\left(P_E+\frac{1}{2}\right)\right)-g\left(h\left(\frac{1}{2}\right) \right)\nonumber\\
 &\geq \frac{1}{P_E+\frac{1}{2}}  \frac{f\left(P_A+\frac{1}{2}\right)+h\left(P_E+\frac{1}{2}\right)-2   f\left( \frac{1}{2}\right)  }{f\left( \frac{1}{2} \right)}    \nonumber\\
  &\geq \frac{P_A}{P_E+\frac{1}{2}}  \frac{\str \left(\boldsymbol{M}  \boldsymbol{M}^\top \right)
  + \str \left(\boldsymbol{N}\boldsymbol{N}^\top\right) }{\str\left(\boldsymbol{M}  \boldsymbol{M}^\top+\boldsymbol{N}\boldsymbol{N}^\top\right)}.
\end{align}
By means of Eq.~\eqref{eq:gfgh} we immediately arrive at
\begin{equation}\label{eq:21}
  \chi_{H\ox}
 + \chi_{A\ox}
  \geq \frac{P_A}{P_E+\frac{1}{2}}  \frac{\str \left(\boldsymbol{M}  \boldsymbol{M}^\top \right)+ \str \left(\boldsymbol{N}\boldsymbol{N}^\top\right) }{ \str\left(\boldsymbol{M}  \boldsymbol{M}^\top+\boldsymbol{N}\boldsymbol{N}^\top\right)} .
\end{equation}
From \cite{Caruso_2008},
the canonical block form for $n$-mode quantum Gaussian channels where $\boldsymbol{M}$ is nonsingular is as follows
\begin{equation}\label{canonical}
 {\boldsymbol S}=\left(\begin{array}{cccc}
\1_n & 0 & \1_n-\boldsymbol{J}^{\top} & 0  \\
0 & \boldsymbol{J} & 0 & -\boldsymbol{J} \\
\1_n & 0 & \1_n & 0 \\
0 & \1_n-\boldsymbol{J} & 0 & \boldsymbol{J}
\end{array}\right),\quad or \quad
 {\boldsymbol S}=\left(\begin{array}{cccc}
\1_n & 0 & 0 & \1_n-\boldsymbol{J}^{\top}  \\
0 & \boldsymbol{J} & -\boldsymbol{J} & 0 \\
0 & \1_n & \1_n & 0 \\
\1_n-\boldsymbol{J} & 0 & 0 & \boldsymbol{J}
\end{array}\right),
\end{equation}
where $\boldsymbol{J}$ is a $n\times n$ block-diagonal matrix in the real Jordan form.
By assuming that the eigenvalues of $\boldsymbol{J}$ are different from $1$ together with the degradability condition implies to have
\begin{align}
\boldsymbol{J} \boldsymbol{J}^\top -(\1_n -\boldsymbol{J}^{-1})\boldsymbol{J}
\boldsymbol{J}^\top(\1_n -\boldsymbol{J}^{\top})  &\geq 0, \\
(\1_n -\boldsymbol{J}^{-1})(\1_n -\boldsymbol{J}^{\top}) &\leq  \1_n.
\end{align}
Next, using the fact that, if $\boldsymbol{B}\boldsymbol{A}\leq \1$,
then $\boldsymbol{A} \boldsymbol{B}=\boldsymbol{B}^{-1} (\boldsymbol{B}\boldsymbol{A}) \boldsymbol{B} \leq \1$, we get
\begin{equation}
(\1_n -\boldsymbol{J}^{\top})(\1_n -\boldsymbol{J}^{-1})\leq  \1_n.
\end{equation}
Therefore, we have
\begin{equation}
\str\left(\boldsymbol{M}\boldsymbol{M}^\top+ \boldsymbol{N}\boldsymbol{N}^\top\right) \leq 2 \str\left(\boldsymbol{M}\boldsymbol{M}^\top\right).
\end{equation}
Finally, replacing this in \eqref{eq:21}, we arrive at
\begin{equation}\label{eq:131}
  \chi_{H\ox}
 + \chi_{A\ox}
  \geq \frac{P_A}{P_E+\frac{1}{2}}\frac{\str \left(\boldsymbol{M}  \boldsymbol{M}^\top \right)
  + \str \left(\boldsymbol{N}\boldsymbol{N}^\top\right) }{ 2 \str\left(\boldsymbol{M}
  \boldsymbol{M}^\top\right)} \geq \frac{P_A}{2 P_E+1} .
\end{equation}
\end{proof}

\medskip
\begin{remark}
Relaxing the requirements that $\boldsymbol{M}$ be non-singular together with
${\rm rank}(\boldsymbol{N}\boldsymbol{\Sigma}_{2N}\boldsymbol{N}^\top)=2N$, we may have the following:

\begin{itemize}
\item
Relation \eqref{eq:uncertainty} still holds if $\boldsymbol{M}$ and $\boldsymbol{N}$ have diagonal form such that
$$
{\rm str}\left(\boldsymbol{M}\boldsymbol{M}^\top\right)
+{\rm str}\left(\boldsymbol{N}\boldsymbol{N}^\top\right)
={\rm str}\left(\boldsymbol{M}\boldsymbol{M}^\top
+\boldsymbol{M}\boldsymbol{M}^\top\right).
$$

\item
Relation \eqref{eq:uncertainty} still holds if $\boldsymbol{M}\boldsymbol{M}^\top
\leq \boldsymbol{N}\boldsymbol{N}^\top$ or $\boldsymbol{M}\boldsymbol{M}^\top
\geq \boldsymbol{N}\boldsymbol{N}^\top$.

\item
A bound tighter than Eq.~\eqref{eq:uncertainty} exists
if the $N$-mode quantum channel results as the tensor product of identical single mode Gaussian quantum channels (see Appendix \ref{app:OMB}).

\end{itemize}
\end{remark}

\medskip
In conclusion, unless one of the two energy constraints $P_A$ and $P_E$
is zero, the sum of the classical capacities with helper is always strictly
greater than zero.
On the other hand, if one of $P_A$ or $P_B$ is zero, the identity or
the SWAP unitary show that it can happen that both capacities are zero.


\subsection{Conferencing encoders}

Here we consider conferencing encoders, that is a situation where
Alice and the helper can freely communicate classical messages, to prepare
signal states for the transmission of a common message.
The classical capacity with conferencing encoders is then defined
in such a way that the encoders (Alice and the helper) are restricted
to use product states between $A$ and $E$.

An encoding CPTP map ${\cal E} : M \to  {\cal T} (A^n)\otimes {\cal T}(E^n)$ can be thought of as
two local encoding maps performed by Alice and Helen, respectively, and given by
${\cal E}_A : M \to {\cal  T}(A^n)$ and ${\cal E}_H: M \to{\cal T}(E^n)$. These can be realized by preparing pure product states $\{|\alpha_m\rangle\otimes |\eta_m\rangle\}$ to
be input across $A^n$ and $E^n$ of $n$ instances of the channel.
A decoding CPTP map ${\cal D}: {\cal T}(B^n) \to M$
can be realized by a POVM $\{\Lambda_m\}$.
The probability of error for a particular message $m$ is
\begin{equation}
P_e(m)=1-\tr\left(\Lambda_m {\cal N}^{\otimes n} \left(\alpha_m^{A^n}\otimes \eta_m^{E^n}\right)
\right).
\end{equation}

\begin{definition}
A classical code for conferencing encoders of block length $n$ is a family of triples
$\{|\alpha_m\rangle^{A^n}, |\eta_m\rangle^{E^n},\Lambda_m\}$ with the error probability
$\overline{P}_e :=\frac{1}{|M|}\sum_mP_e(m)$ and rate $\frac{1}{n}\ln |M|$. A rate $\tt R$
is achievable if there is a sequence of codes over their block length $n$ with $\overline{P}_e$ converging to 0 and rate converging to $\tt R$. The classical capacity with conferencing encoders of $W$, denoted by $C_{\text{\Phone}}(W)$ is the maximum achievable rate. If the sender and helper
are restricted to fully separable states $\alpha^{A^n}_m$ and $\eta^{E^n}_m$ , i.e., convex combinations of tensor products $\alpha^{A^n}_m = \alpha^{A^1}_{1m}\otimes \ldots\otimes
\alpha^{A^n}_{nm}$ and $\eta^{E^n}_m =\eta^{E^1}_{1m}\otimes \ldots\otimes
\eta^{E^n}_{nm}$, for all $m$, the largest achievable rate is
denoted by $C_{\text{\Phone}\otimes}(W)$ and is henceforth referred to as classical capacity with product
conferencing encoders.
\end{definition}

\begin{theorem}\label{thm:conf}
For a Gaussian isometry $W : AE \to BF$, satisfying the condition
 the classical capacity with
conferencing encoders is given by
\begin{equation}
  C_{\text{\Phone}}(W,P_A,P_E)
    = \sup_n \max_{  \{ p(x^n),  \,  \alpha_{x^n}^{A^n}\otimes\eta_{x^n}^{E^n} \}  }
    \frac{1}{n}\chi \left(\left\{ p(x^n),
    {\cal N}^{\otimes n}(   \alpha_{x^n}^{A^n}\otimes\eta_{x^n}^{E^n}  )\right\} \right),
\end{equation}
where the maximization is over ensembles
respecting energy constraints
$\sum_{x^n} p(x^n) \tr(\alpha_{x^n}^{A^n}H_{A^n})\leq nP_A$
and $\sum_{x^n} p(x^n) \tr(\eta_{x^n}^{E^n}H_{E^n})\leq nP_E$.

Similarly, the product state capacity of conferencing encoders is given by the formula,
\begin{equation}
  C_{\text{\Phone}\ox}(W,P_A,P_E)
    = \max_{  \{ p(x),  \,  \alpha_{x}^{A}\otimes\eta_{x}^{E} \}  }
    \chi \left(\left\{ p(x), {\cal N}(\alpha_{x}^{A}\otimes\eta_{x}^{E})\right\} \right),
\end{equation}
where the maximization is over ensembles
respecting energy constraints
$\sum_{x} p(x) \tr(\alpha_{x}^{A}H_{A})\leq P_A$
and $\sum_{x} p(x) \tr(\eta_{x}^{E}H_{E})\leq P_E$.
\end{theorem}

\begin{proof}
The direct part, i.e. the ``$\geq$" inequality, follows from the HSW Theorem \cite{SchW97,H98}.
For the converse part, i.e. the ``$\leq$" inequality, the proof goes like that of \cite[Thm.~4]{KSWY2}.
\end{proof}

\medskip
A lower bound on the classical capacity with conferencing encoders
follows from the uncertainty relation of Theorem \ref{thm:uncertainty}.
In fact from the definition of the conferencing encoder, we obtain directly
\begin{equation}
C_{\text{\Phone}\ox} \geq \max\left\{ \chi_{H\ox}(W), \chi_{A\ox}(W) \right\},
\end{equation}
and thus
\begin{equation}
C_{\text{\Phone}\ox} \geq \frac{ \chi_{H\ox}(W) + \chi_{A\ox}(W)}{2}
                 \geq \frac{1}{2} \frac{ \min\{P_A, P_E\}}{2\max\{P_E, P_A\}+1}.
\end{equation}
In other words, the classical capacity with conferencing encoders is always positive,
provided the energy is non-zero on both inputs.

\medskip
Consider a symplectic transformation $\boldsymbol{S}$, given in the
block form Eq. \eqref{sympl}. Consider seed states with covariance
matrices ${\boldsymbol{V}}_A$ and ${\boldsymbol{ V}}_E$ with zero vector mean.
Suppose the classical message is encoded by applying displacement operator
to the seed states. We assume that the distribution of the classical messages
is a Gaussian distribution with zero mean whose covariance matrix is given by
$\boldsymbol{V}_{mod}$. The action of encoding is described as follows:
\begin{equation}
\begin{split}
\overline{\boldsymbol{V}}_A &= {\boldsymbol{V}}_A+{\boldsymbol{V}}_{mod},\\
\overline{\boldsymbol{V}}_E &= {\boldsymbol{V}}_E+{\boldsymbol{V}}_{mod}.
\end{split}
\end{equation}
The covariance matrices of the output state and the output averaged state are labelled
${\boldsymbol{V}}_B$ and $\overline{\boldsymbol{V}}_B$ respectively and given by
\begin{equation}
\begin{split}
{\boldsymbol{ V}}_B&=\boldsymbol{M}{\boldsymbol{ V}}_A\boldsymbol{M}^\top
+\boldsymbol{N}{\boldsymbol{ V}}_E\boldsymbol{N}^\top,\\
\overline{\boldsymbol{ V}}_B&=\boldsymbol{M}\overline{\boldsymbol{ V}}_A\boldsymbol{M}^\top
+\boldsymbol{N}\overline{\boldsymbol{ V}}_E\boldsymbol{N}^\top.
\end{split}
\end{equation}
Let us evaluate the transmission of classical information by conference encoders using the seed states
${\boldsymbol{ V}}_A = {\boldsymbol{ V}}_E = \boldsymbol{I}/2$ and ${\boldsymbol{ V}}_{mod} = c\boldsymbol{I}/2$.

Imposing the input energy constraint we have (assuming that Alice and the helper are bounded by same energy) in terms of covariance matrices:
\begin{equation}
\frac{\tr\overline{\boldsymbol{ V}}_A}{2n}\leq P_A+\frac{1}{2}.
\end{equation}

Choosing $c = 2 P_A$
we get the Holevo function of this ensemble to be
\begin{equation}
\sum_{i=1}^n\left[
g\left(\frac{(2P_A+1)\nu_i-1}{2}\right)-g\left(\frac{\nu_i-1}{2}\right)
\right],
\end{equation}
where $\nu_i$ are the symplectic eigenvalues of $\boldsymbol{M}\boldsymbol{M}^\top
+ \boldsymbol{N}\boldsymbol{N}^\top$. As $g$ is concave monotonic in the argument we have
the above quantity non-zero whenever $P_A> 0$.
In particular, for the case of beam-splitter, amplifier and conjugate amplifier,
$\boldsymbol{M}\boldsymbol{M}^\top + \boldsymbol{N}\boldsymbol{N}^\top = \boldsymbol{I}$,
we have the classical
information transmission for the above setting given by $g(P_A)$,
which is the transmission of ideal channel with mean photon number $P_A$.


\section{Continuity of capacities in communication assisted by helper}
\label{sec:continuity}

The quantum and classical capacities assisted by separable helper that we
defined and studied above also satisfy uniform continuity.

\begin{theorem}
For input  and output energy-limited Gaussian
channels ${\cal N}^{AE\to B}$ and ${\cal M}^{AE\to B}$, if $\|{\cal N}^{AE\to B}
-{\cal M}^{AE\to B}\|_{\diamond} \leq 2\epsilon$, then
\begin{align}
|C_{H\ox}({\cal N})-C_{H\ox}({\cal M})| &\leq 28 \sqrt{\epsilon}\,
S\left( \gamma_B\left(\frac{4 P_B}{\sqrt{\epsilon}}\right)\right)+3 g\left(\sqrt{\epsilon}+\frac{1}{2}\right)\\
|Q_{H\ox}({\cal N})-Q_{H\ox}({\cal M})| &\leq 28 \sqrt{\epsilon}\,
S\left( \gamma_B\left(\frac{4 P_B}{\sqrt{\epsilon}}\right)\right)+3
g\left(\sqrt{\epsilon}+\frac{1}{2}\right),
\end{align}
where $g$ is given in Eq. \eqref{eq:gfunction} and
$\gamma_X(P)$ is the Gibbs state of system $X$.
\end{theorem}

\begin{proof}
The proof immediately follows from \cite[Thm.~9]{W17} by noticing that
\begin{equation}
\|{\cal N}^{A\to B}_\eta
-{\cal M}^{A\to B}_\eta\|_{\diamond} \leq
\|{\cal N}^{AE\to B}
-{\cal M}^{AE\to B}\|_{\diamond}\leq 2\varepsilon, \quad \forall \eta\in E
\end{equation}
the channels ${\cal N}^{A\to B}_\eta$, ${\cal M}^{A\to B}_\eta$ being
restrictions of ${\cal N}^{AE\to B}$, ${\cal M}^{AE\to B}$ respectively.
Furthermore, for any $\eta\in E$, the energy limitation for the output state,
according to Eq. \eqref{eq:cohamiltonian} of Lemma \ref{lem:lemma2},
will be as follows
\begin{equation}
  \Tr \mathcal{N}_\eta(\rho)H_{B} = \Tr \mathcal{N}^{AE\to B}(\rho\otimes \eta)H_B  \leq 2P_A+2P_E \equiv P_B,
\end{equation}
thus concluding the proof.
\end{proof}

\begin{remark}
If we take $\boldsymbol{d}_\rho=\boldsymbol{d}_\eta=0$ in Lemma \ref{lem:lemma2},
then we have $\Tr \mathcal{N}_\eta(\rho)H_{B} \leq P_A+c c_E P_E$.
By choosing $c>0$ such that $c c_E\leq \alpha$ together with
$H_B=c \hat{\boldsymbol{r}}_B\hat{\boldsymbol{r}}_B^\top$, the quantity
$P_B=P_A + \alpha P_E$ plays the role of $\tilde{E}=\alpha E + E_0$ in \cite{W17}.
\end{remark}


\section{Conclusion}
\label{sec:conclusion}

We have created a model of communication via infinite-dimensional channels
defined by a bipartite unitary, when assisted by a passive helper in the
environment. In this model, we have investigated quantum and classical
capacities, proving various general capacity theorems, the former without and
and with energy constraints, the latter with energy constraints, with respect
to natural assumptions on the Hamiltonians involved.

In particular, in Bosonic Gaussian systems, where the Hamiltonian is that
of several quantum harmonic oscillators and with a Gaussian unitary defining
the interaction, we showed that the capacity formulas lead to simple
expressions, when the helper is restricted to Gaussian states. Furthermore,
for the classical capacity we showed a tradeoff (``uncertainty'') relation between
the capacity of Alice assisted by the helper, and that of the helper assisted
by Alice in terms of the respective input powers, and a lower bound on the
classical capacity with conferencing encoders Alice and helper.

Practically all of our general capacity formulas are multi-letter, and
it remains to find bipartite unitaries for which any of them is both
non-trivial and explicitly computable, or at least a single-letter formula.
In that respect, although we proved the impossibility of having a universally
(anti-)degradable Gaussian unitary, it remains open the possibility that
for every environment state $\eta$ the effective channel ${\cal N}_\eta$
has a well defined degradability property (not the same for all $\eta$).
More generally, we would like to know unitaries that are universally
degradable (not just for Gaussian helper inputs), for a single-letter quantum
capacity, and likewise unitaries resulting in universally additive channels
for the Holevo capacity.
The lower bound on conferencing encoders based on the capacity uncertainty
relation seems very weak, and it remains open to prove better bounds.

Finally, it could be interesting to turn the role of helper into that of an adversary and study
how the quantum communication capabilities between Alice and Bob will be hampered
by this adversary and its energy power.
In this sense the presented model paves the way to investigate arbitrarily varying quantum channels also in infinite dimensional spaces. A topic of particular relevance for the secrecy of practical (in fiber and free space) quantum communication.


\section*{Acknowledgments}
SM and AW acknowledge fruitful discussions with Siddharth Karumanchi
in the early stages of the present project.
AW furthermore thanks I. Thompson and C. Dexter for insights into
the influence of the environment on quantum and classical systems.

SM acknowledges the financial support of the Horizon-2020 Programme of
the European Commission, under the FET-Open grant agreement QUARTET, number 862644.
AW acknowledges financial support by the Spanish MINECO
(projects FIS2016-86681-P and PID2019-107609GB-I00/AEI/10.13039/501100011033)
with the support of FEDER funds, and the Generalitat de Catalunya
(project 2017-SGR-1127).


\appendix

\section{Proof of Theorem \ref{thm:nongaussian}}
\label{app:thmproof}

The proof is divided into two parts, one concerning the case $q<1$ and the another the case $q>1$.
In the former, for $\frac{1}{\sqrt{2}} \leq q < \frac{1}{2}+\frac{\sqrt{3}}{6}$,
we obtain  a special convex combination of quantum density matrices which has an image
through $\Gamma$  with some negative eigenvalues.  Since the method for other cases within
the interval $1/2 \leq q<1$ is similar, we just numerically show the negativity of some
eigenvalues for the images through $\Gamma$ of special convex combinations of density matrices.
In contrast, the case $q>1$ is different, as a contradiction is achieved based on the fact
that the quantum relative entropy cannot increase by quantum operations.

\subsection{Case $q<1$}
\label{subsec:q<1}

\begin{proposition}
\label{prop}
The two-mode Gaussian unitaries $U^{(q)}$ for $\frac{1}{\sqrt{2}} \leq q < \frac{1}{2}+\frac{\sqrt{3}}{6}$
are neither universally degradable, nor universally anti-degradable.
\end{proposition}

\begin{proof}
It is enough to prove that there exists a state $\eta_E$ for which the channel
${\cal N}^{A\to B}_{\eta_E}$ is anti-degradable. In view of Remark \ref{obs1},
this state is necessarily non-Gaussian.

The $U^{(q)}$ corresponding to \eqref{Scan1} turns out to be
\begin{equation}
U^{(q)}=e^{\arccos\sqrt{q}\, {\left(\hat{a}^\dagger \hat{b} -\hat{a} \hat{b}^\dagger\right)}},
\end{equation}
for $q\in(0,1)$. Then, for the Fock state $  \vert n\rangle \vert 1\rangle$, we have
\begin{equation}
  U^{(q)}  \vert n\rangle \vert 1\rangle =-\frac{1}{\sqrt{(n+1)(1-q)}}\sum_{\ell=0}^{n+1} (-1)^\ell \sqrt{\binom{n+1}{\ell}} (1-q)^{\ell/2} q^{\frac{n-\ell}{2}}((n+1)(1-q)-\ell) \vert n+1-\ell\rangle\vert \ell\rangle.
\end{equation}
By selecting $n=0, 1$, we get
\begin{equation}\label{eq:68}
U^{(q)}  \vert 0\rangle \vert 1\rangle = -\frac{1}{\sqrt{1-q}}\left(  (1-q) \vert 1\rangle \vert 0\rangle+ \sqrt{q(1-q)} \vert 0\rangle \vert 1\rangle\right),
\end{equation}
and
\begin{equation}
U^{(q)}  \vert 1\rangle \vert 1\rangle =-\frac{1}{\sqrt{2(1-q)}}\left(2\sqrt{q}(1-q) \vert 2\rangle \vert 0\rangle -\sqrt{2}\sqrt{1-q}(1-2q)\vert 1\rangle \vert 1\rangle -2(1-q)\sqrt{q} \vert 0\rangle \vert 2\rangle \right).
\end{equation}

Consider now the channel with environment in the Fock state $|1\rangle\langle 1|$, i.e.
\begin{equation}
\mathcal{N}_q(\rho)=\Tr_E\left( U^{(q)} (\rho\otimes \vert 1\rangle \langle 1\vert){U^{(q)}}^\dagger \right).
\end{equation}
Let us assume that there exists a channel $\Gamma$ such that
\begin{equation}
\label{degrad}
\Gamma\circ\mathcal{N}(\rho)=\widetilde{\mathcal{N}}(\rho).
\end{equation}
Inputting $\rho= \vert 0\rangle \langle 0\vert$, we find that
\begin{equation}
\mathcal{N}_q( \vert 0\rangle \langle 0\vert)= q  \vert 0\rangle\langle 0\vert   +(1-q)  \vert 1\rangle\langle 1\vert,
\end{equation}
and
\begin{equation}
\widetilde{\mathcal{N}}_q( \vert 0\rangle \langle 0\vert)=q  \vert 1\rangle\langle 1\vert   +(1-q)  \vert 0\rangle\langle 0\vert.
\end{equation}
 Therefore, according to \eqref{degrad}, we should have
\begin{equation}\label{eq:11}
 q  \Gamma(\vert 0\rangle\langle 0\vert)   +(1-q) \Gamma( \vert 1\rangle\langle 1\vert)= q  \vert 1\rangle\langle 1\vert   +(1-q)  \vert 0\rangle\langle 0\vert.
\end{equation}

Analogously, inputting $\rho= \vert 1\rangle \langle 1\vert$, we get
\begin{equation}
\mathcal{N}_q( \vert 1\rangle \langle 1\vert)=\frac{1}{2-2q}\left( 4q (1-q)^2 \vert 0\rangle \langle 0\vert  +2 (1-q) (1-2q)^2 \vert 1\rangle \langle 1\vert + 4q (1-q)^2 \vert 2\rangle \langle 2\vert \right),
\end{equation}
and
\begin{equation}
\widetilde{\mathcal{N}}_q( \vert 1\rangle \langle 1\vert)=\frac{1}{2-2q}\left(  4q (1-q)^2 \vert 2\rangle \langle 2\vert  +2 (1-q) (1-2q)^2  \vert 1\rangle \langle 1\vert + 4q (1-q)^2 \vert 0\rangle \langle 0\vert \right).
\end{equation}
Hence, according to \eqref{degrad}, we should have
\begin{multline}
\label{Gamma1}
\frac{1}{2-2q}\left( 4q (1-q)^2  \Gamma(\vert 0\rangle \langle 0\vert)  +2 (1-q) (1-2q)^2\Gamma(\vert 1\rangle \langle 1\vert) +4q (1-q)^2\Gamma(\vert 2\rangle \langle 2\vert) \right)=\\
\frac{1}{2-2q}\left(  4q (1-q)^2   \vert 2\rangle \langle 2\vert  +2 (1-q) (1-2q)^2 \vert 1\rangle \langle 1\vert +4q (1-q)^2\vert 0\rangle \langle 0\vert \right).
\end{multline}

Now, from \eqref{eq:11}, we derive
\begin{equation}
\Gamma(|1\rangle\langle 1|)= \frac{q}{1-q} \vert 1\rangle\langle 1\vert   +  \vert 0\rangle\langle 0\vert- \frac{q}{1-q}   \Gamma(\vert 0\rangle \langle 0\vert),
\end{equation}
which, inserted into \eqref{Gamma1}, yields
\begin{multline}
 q\left(1-2q^2\right)  \Gamma(\vert 0\rangle \langle 0\vert) +2q(1-q)^2 \Gamma(\vert 2\rangle \langle 2\vert) = \\
  2q (1-q)^2 \vert 2\rangle \langle 2\vert  +  (1-2q)^3 \vert 1\rangle \langle 1\vert +  (1-q)(-1+6q-6q^2)\vert 0\rangle \langle 0\vert.
\end{multline}
 Isolating the term $\Gamma(\vert 2\rangle \langle 2\vert)$ at l.h.s., we arrive at
\begin{equation}\label{Gamma2}
  \Gamma(\vert 2\rangle \langle 2\vert)=  -\frac{ 1-2q^2}{2(1-q)^2} \Gamma(\vert 0\rangle \langle 0\vert)+\vert 2\rangle \langle 2\vert+\frac{(1-2q)^3}{2q(1-q)^2}\vert 1\rangle \langle 1\vert +
  \frac{-1+6q-6q^2}{2q(1-q)} \vert 0\rangle \langle 0\vert.
\end{equation}

 At this point, taking a convex combination of $\Gamma(\vert 1\rangle \langle 1\vert)$
and $\Gamma(\vert 2\rangle \langle 2\vert)$ must give a positive operator, given that $\Gamma$ is a CPTP map. Consider then
\begin{equation}
 \frac{1-q}{q}\Gamma(\vert 1\rangle \langle 1\vert)+\frac{2(1-q)^2}{ 2q^2-1}\Gamma(\vert 2\rangle \langle 2\vert),
\end{equation}
with ${q \geq \frac{1}{\sqrt{2}}}$,
we get
\begin{align}
 \frac{1-q}{q}\Gamma(\vert 1\rangle \langle 1\vert)+\frac{2(1-q)^2}{ 2q^2-1}\Gamma(\vert 2\rangle \langle 2\vert) =&
  \frac{2(1-q)^2}{ 2q^2-1}\vert 2\rangle \langle 2\vert \notag\\
  & +\left[1+ \frac{(1-2q)^3}{q(2q^2-1)} \right]\vert 1\rangle \langle 1\vert  \notag\\
  & + \left[ \frac{1-q}{q}+\frac{(1-q)(-1+6q-6q^2)}{q(2q^2-1)}\right]\vert 0\rangle \langle 0\vert.
  \end{align}
 Now, if we analyze the coefficients at r.h.s. (which correspond to the eigenvalues of the
convex combination of $\Gamma(\vert 1\rangle \langle 1\vert)$
and $\Gamma(\vert 2\rangle \langle 2\vert)$) we have
\begin{align}
\frac{2(1-q)^2}{ 2q^2-1} &\geq 0 \qquad \text{for}
\qquad {\frac{1}{\sqrt{2}} \leq q < 1}, \\
  1+ \frac{(1-2q)^3}{q(2q^2-1)} &< 0 \qquad \text{for} \qquad
  {\frac{1}{\sqrt{2}} \leq q < \frac{1}{2}+\frac{\sqrt{3}}{6}},  \\
  \frac{1-q}{q}+\frac{(1-q)(-1+6q-6q^2)}{q(2q^2-1)}
  &{ > 0 \qquad \text{for} \qquad
  \frac{1}{\sqrt{2}} \leq q < 1}.
  \end{align}
 Thus we can conclude that the channel $\Gamma$ does not exist (at least for
{$\frac{1}{\sqrt{2}} \leq q < \frac{1}{2}+\frac{\sqrt{3}}{6}$)} because its eigenvalues should have been positive. This in turn means that in the above range of $q$ values the Gaussian unitaries are neither universally degradable nor universally anti-degradable.
\end{proof}

\begin{remark}
\label{re:prop2}
Numerical investigations (see below) suggests that the statement of
Proposition \ref{prop} holds actually true for $q$ between $1/2$ and $1$.
\end{remark}

Let us consider the convex combination of states as
\begin{equation}
  \label{eq:coefficient}
  \frac{k_m}{k_m+k_n}\Gamma(\vert n\rangle\langle n\vert)+ \frac{ k_n}{k_m+k_n}\Gamma(\vert m\rangle\langle m\vert),
\end{equation}
where  $k_m$ and  $k_n$ are the coefficients in front of $\Gamma(\vert 0\rangle\langle 0\vert)$
for the expressions of $\Gamma(\vert m\rangle\langle m\vert)$ and $\Gamma(\vert n\rangle\langle n\vert)$, respectively.

Define  $c(k_n, k_m)$ the coefficient of $\vert 1\rangle\langle 1\vert$ for the
combination \eqref{eq:coefficient}.
Figures \ref{fig1} and \ref{fig2} show that there is always a negative $c(k_n, k_m)$
for $q\in [\frac{1}{2},1)$.

\begin{figure}[H]
	\centering
	\includegraphics[width=0.4\textwidth]{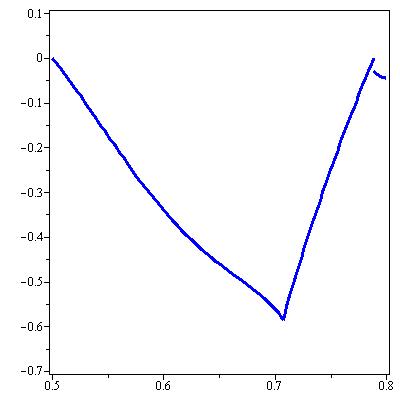}
		\caption{ Quantities $c(k_n, k_m)$ vs $q$. In particular, in the range
		 $[1/2, 1/\sqrt{2}]$ it is plotted $c(k_4, -k_2)$. In the range $[ 1/\sqrt{2}, 1/2 +\sqrt{3}/6)$ it is plotted $c(k_2, -k_1)$, according to Proposition \ref{prop}.
		 Finally, in the range $[1/2 +\sqrt{3}/6, 0.8]$ it is plotted $c(-k_4, k_2)$. In the point $q=1/2 +\sqrt{3}/6$, it is $c(k_2, -k_1)=0$ while $c(-k_4, k_2)=-0.0303$. }
	\label{fig1}
\end{figure}

\begin{figure}[H]
	\centering
	\includegraphics[width=0.4\textwidth]{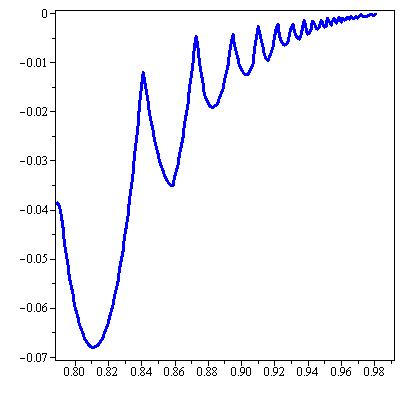}
		\caption{  The quantity $\min_{n,m \leq 50} C( \vert  k_n\vert,  \vert k_m\vert)$
		vs $q$ when $k_n k_m<0$.}
	\label{fig2}
\end{figure}


\subsection{Case $q>1$}
\label{prop3}

The $U^{(q)}$ corresponding to \eqref{Scan2} turns out to be
\begin{equation}
  U^{(q)}= e^{ i\,{\rm arccosh}\sqrt{q} \left(\hat{a}^\dagger \hat{b}^\dagger + \hat{a}\hat{b}\right)}.
\end{equation}
Using the disentangling formula for the $SU(1,1)$ group, it is possible to rewrite it as
\begin{equation}\label{eq:Usplitted}
U^{(q)}=e^{r \hat{a}^\dagger \hat{b}^\dagger} e^{-s\left(\hat{a}^\dagger \hat{a} +\hat{b} \hat{b}^\dagger\right)} e^{r \hat{a}\hat{b}},
\end{equation}
where
\begin{equation}
r=i\,\sqrt{\frac{q-1}{q}}, \quad s= \ln{\sqrt{q} }.
\end{equation}
Let us now compute the action of $U^{(q)}$ on the Fock state $\vert m 1\rangle$. It results
\begin{eqnarray}
 U^{(q)}\vert m 1\rangle &=&   e^{\hat{a}^\dagger \hat{b}^\dagger r} e^{-\left(\hat{a}^\dagger \hat{a}
 +\hat{b} \hat{b}^\dagger\right) s } \left(\sum_{n=0}^{\infty} \frac{ r^n \left(\hat{a}\hat{b}\right)^n}{n!}     \vert m 1\rangle\right) \\
   &=& e^{\hat{a}^\dagger \hat{b}^\dagger r} e^{-\left(\hat{a}^\dagger \hat{a} +\hat{b} \hat{b}^\dagger\right) s }\left( \vert m 1\rangle+\sqrt{m} r\vert(m-1)0\rangle  \right)\\
   &=&    e^{\hat{a}^\dagger \hat{b}^\dagger r} \left(\sum_{n=0}^{\infty} \frac{(-1)^n s^n \left(
   \hat{a}^\dagger \hat{a} +\hat{b} \hat{b}^\dagger\right)^n}{n!}   \right)\left(\vert m 1\rangle+\sqrt{m} r\vert(m-1)0\rangle \right)\\
   &=&   e^{\hat{a}^\dagger \hat{b}^\dagger r}    \left(\sum_{n=0}^{\infty} \frac{(-1)^n s^n (m+2)^n}{n!} \vert m 1\rangle +\sqrt{m}r\sum_{n=0}^{\infty} \frac{(-1)^n s^n (m-1+1)^n}{n!}   \vert(m-1)0\rangle\right)   \\
   &=& e^{\hat{a}^\dagger \hat{b}^\dagger r}    \left(e^{-(m+2)s}  \vert m 1\rangle +\sqrt{m}re^{-ms}   \vert(m-1)0\rangle     \right) \\
   &=&   e^{-(m+2)s}  \sum_{n=0}^{\infty}  \sqrt{n+1}\sqrt{\binom{n+m}{m}} r^n\vert (n+m)(n+1)\rangle\notag\\
    &&+\sqrt{m}re^{-ms}  \sum_{n=0}^{\infty} \sqrt{\binom{n+m-1}{m-1}} r^n\vert(n +m -1)n\rangle\\
    &=&\sqrt{m}re^{-ms}   \vert(m-1)0\rangle  \notag  \\
     &&
     +\sum_{n=0}^{\infty} \left( e^{-(m+2)s} \sqrt{n+1}\sqrt{\binom{n+m}{m}} +\sqrt{m}r^2 e^{-ms} \sqrt{\binom{n+m}{m-1}} \right)r^{n}\vert(n+m)(n+1)\rangle. \nonumber\\
      \end{eqnarray}
Then, we can get
  \begin{multline}\label{eq:specN}
   \mathcal{N}(\vert m\rangle\langle m\vert)= m \vert r\vert^2 e^{-2ms} \vert m-1\rangle\langle m-1\vert\\
    +  \sum_{n=0}^{\infty}\left( e^{-(m+2)s} \sqrt{n+1}\sqrt{\binom{n+m}{m}}-\sqrt{m}\frac{q-1}{q} e^{-ms} \sqrt{\binom{n+m}{m-1}} \right)^2 \vert r\vert^{2n}\vert n+m\rangle\langle n+m\vert,
      \end{multline}
and
  \begin{multline}\label{eq:speccalN}
 \widetilde{\mathcal{N}}(\vert m\rangle\langle m\vert)= m \vert r\vert^2 e^{-2ms} \vert 0\rangle\langle 0\vert\\
  +    \sum_{n=0}^{\infty}\left( e^{-(m+2)s} \sqrt{n+1}\sqrt{\binom{n+m}{m}}-\sqrt{m}\frac{q-1}{q} e^{-ms} \sqrt{\binom{n+m}{m-1}}  \right)^2 \vert r\vert^{2n}\vert n+1\rangle\langle n+1\vert.
      \end{multline}

\bigskip
It is known that for any completely  positive map $\Gamma$ and two density matrices
$\rho$ and $\sigma$,
the following inequality for quantum relative entropy holds true (contractive property)
\begin{equation}
 D(\Gamma(\rho)\Vert \Gamma(\sigma))\leq D(\rho\Vert \sigma).
\end{equation}
By assuming the degradability condition for ${\cal N}$, we should have
\begin{equation}\label{eq:Dineq}
   D\left(\widetilde{\cal N}(\vert m_1\rangle\langle m_1\vert)\Big\Vert \widetilde{\cal N}(\vert m_2\rangle\langle m_2\vert)\right)
   \leq D\left({\cal N}(\vert m_1\rangle\langle m_1\vert)\big\Vert{\cal  N}(\vert m_2\rangle\langle m_2\vert)\right),
\end{equation}
for all $ m_1> m_2\in\mathbb{N}$.
From Eqs.~\eqref{eq:specN} and \eqref{eq:speccalN}, we have
\begin{equation}
D\left({\cal N}(\vert m_1\rangle\langle m_1\vert)\big\Vert{\cal  N}(\vert m_2\rangle\langle m_2\vert)\right)=m_1 \vert r\vert^2e^{-2m_1 s} \ln\frac{m_1 \vert r\vert^2e^{-2m_1 s} }{c^q_{m_2 (m_1- m_2-1)} }+\sum_{n=0}^{\infty} c^q_{m_1 n}\ln\frac{c^q_{m_1 n}}{c^q_{m_2 (m_1-m_2+n)}},
\end{equation}
and
\begin{equation}
D\left(\widetilde{\cal N}(\vert m_1\rangle\langle m_1\vert)\Big\Vert \widetilde{\cal N}(\vert m_2\rangle\langle m_2\vert)\right) =m_1 \vert r\vert^2e^{-2m_1 s} \ln\frac{m_1 \vert r\vert^2e^{-2m_1 s} }{m_2 \vert r\vert^2e^{-2m_2 s} }+\sum_{n=0}^{\infty} c^q_{m_1 n}\ln\frac{c^q_{m_1 n}}{c^q_{m_2 n}},
\end{equation}
where, according to \eqref{eq:specN} and \eqref{eq:speccalN}, we have defined
\begin{equation}
  \label{eq:cmndef}
  c_{mn}^{q}
   :=\left( e^{-(m+2)s} \sqrt{n+1}\sqrt{\binom{n+m}{m}}
                        -\sqrt{m}\frac{q-1}{q} e^{-ms} \sqrt{\binom{n+m}{m-1}} \right)^2 \vert r\vert^{2n}.
\end{equation}
By simple calculations, we get
\begin{equation}\label{eq:cmndefs}
  c_{mn}^q:=\frac{(n+1-m(q-1))^2}{(n+1)q^{m+2}}\binom{n+m}{m} \left(\frac{q-1}{q}\right)^n.
\end{equation}
Then, Eq.~\eqref{eq:Dineq} reads
\begin{align}
  \label{eq:60}
  m_1 \vert r\vert^2e^{-2m_1 s} \ln\frac{m_1 \vert r\vert^2e^{-2m_1 s} }{m_2 \vert r\vert^2e^{-2m_2 s} }+\sum_{n=0}^{\infty} c_{m_1 n}^{{q}}\ln\frac{c_{m_1 n}^{{q}}}{c_{m_2 n}^{{q}}}
  &\leq   m_1 \vert r\vert^2e^{-2m_1 s}
  \ln\frac{m_1 \vert r\vert^2e^{-2m_1 s} }{c_{m_2 (m_1-m_2-1)}^{{q}} }\nonumber\\
  &+\sum_{n=0}^{\infty} c_{m_1 n}^{{q}}\ln\frac{c_{m_1 n}^{{q}}}{c_{m_2 (m_1-m_2+n)}^{{q}}},
\end{align}
or more simply
\begin{equation}
  \label{eq:61}
  m_1 \vert r\vert^2e^{-2m_1 s} \ln\frac{m_2 \vert r\vert^2e^{-2m_2 s} }{c_{m_2 (m_1 - m_2-1)}^{{q}} }+\sum_{n=0}^{\infty} c_{m_1 n}^{{q}}\ln\frac{c_{m_2 n}^{{q}}}{c_{m_2 (m_1 - m_2+n)}^{{q}}}\geq 0.
\end{equation}
However this inequality can be violated.
In fact, it happens that $c_{m_2 n}^{{q}}=0$ when
\begin{equation}
   n=m(q-1)-1.
\end{equation}
It is then clear that \eqref{eq:61} may be violated when $q$ is close to integer numbers.
More precisely the following result holds true.

\begin{theorem}
For an arbitrary $q>1$, there exist integers $m_1 > m_2$ such that
Eq.~\eqref{eq:61} is not true.
\end{theorem}

\begin{proof}
Let us  consider a fixed rational number $q=\frac{x}{y}>1$.
By selecting $m_2=y$ and $n'=x-y-1$, we have
\begin{equation}
n'=m_2 (q-1)-1,
\end{equation}
that,  from   \eqref{eq:cmndefs}, guarantees $c_{m_2 n'}^{{q}}=0$. On the other hand, if  there exists  $n''$ such that $c_{m_2 ( n''+m_1 - m_2)}^{{q}}=0$ for such $q$,  then we should have
\begin{equation}
  \frac{n''+m_1 - m_2+1}{m_2}=q-1 \quad \Rightarrow \quad     \frac{n''+m_1 - m_2+1}{y}=\frac{x-y}{y}.
\end{equation}
We choose $m_1$ such that $m_1  - y\neq x-y- n''-1$ for any integer $n''=0, 1,\ldots,$  in order to have
$c_{m_2 ( n''+m_1 - m_2)}^{{q}}=0$. We also have $c_{m_1 n'}^{{q}}\neq 0$ to ensure that
\begin{equation}
   c_{m_1 n'}^{{q}}\ln \frac{c_{m_2 n'}^{{q}}}{c_{m_2 (m_1 - m_2+n')}^{{q}}}=-\infty,
\end{equation}
By considering the above descriptions, we  show that relation \eqref{eq:61} violates
for  given small radius $\varepsilon>0$ and  $q< q'< q+ \varepsilon$.
In other words, it is clear that $c_{m_2 n}^{{q'}}\neq 0$, $\forall\; n$ \footnote{
This holds true for $q'$ irrational number.
If $q'$ is a rational number such that $c_{m_2 n}^{ q}=0$, we should have
\begin{equation}
  \frac{k+1}{m_2}=q'-1.
\end{equation}
On the other hand, we have $\vert q-q'\vert \leq \varepsilon$ and hence
\begin{equation}
  \vert q-q' \vert \leq \varepsilon \Rightarrow \left\vert \frac{x}{m_2}-\frac{k+1+ m_2}{m_2}\right\vert \leq \varepsilon,
\end{equation}
which implies that $x=k+m_2+1$ and so $q'=q$. }  and so we have the condition  ${\rm supp}\left({\cal N}(\vert m_1\rangle\langle m_1 \vert)\right) \subseteq {\rm supp}
\left({\cal N}(\vert m_2\rangle\langle m_2 \vert)\right)$.
 Now, for each  $n\geq m_2(q-1)+m_1$, we have
\begin{eqnarray}
  \frac{c_{m_2 n}^{q'}}{c_{m_2 (m_1 - m_2+n)}^{q'}} &=& \frac{     \frac{ (n+1-m_2(q'-1))^2}{(n+1)q'^{m_2+2}}\binom{n+m_2}{m_2} \left( \frac{q'-1}{q'}\right)^n  }{    \frac{ (n+m_1-m_2+1-m_2(q'-1))^2}{(n+m_1-m_2+1)q'^{m_2+2}}\binom{n+m_1-m_2+m_2}{m_2} \left( \frac{q'-1}{q'}\right)^{n+m_1-m_2}     }  \\
  &\leq &  \frac{     \frac{ (n+1-m_2(q'-1))^2}{(n+1)}\binom{n+m_2}{m_2}   }{    \frac{ (n+m_1-m_2+1-m_2(q'-1))^2}{(n+m_1-m_2+1)}\binom{n+m_1}{m_2} \left( \frac{q'-1}{q'}\right)^{m_1-m_2}     }\\
   &\leq& \frac{c_{m_2 n}^{q}}{c_{m_2 (m_1-m_2+n)}^{q}} \left(\frac{ \frac{q-1}{q}}{\frac{q'-1}{q'}}\right)^{m_1-m_2} \label{eq:66}\\
    &\leq&  \frac{c_{m_2 n}^{q}}{c_{m_2 (m_1-m_2+n)}^{q}}.\label{eq:66-1}
\end{eqnarray}
The  \ref{eq:66}  derives from
\begin{equation}
 \frac{n+1-m_2(q'-1)}{n+m_1-m_2+1-m_2(q'-1)} \leq  \frac{n+1-m_2(q-1)}{n+m_1-m_2+1-m_2(q-1)},
 \end{equation}
taking into account that $ n\geq m_2(q-1)+m_1$.

On the other hand, we have
\begin{equation}\label{eq:64}
  \lim_{n\to\infty} \frac{c_{m_2 n}^{ {q}}}{c_{m_2 (m_1-m_2+n)}^{{q}}}=\frac{1}{\vert r\vert^{2m_1}}=\left(\frac{{q}}{{q}-1}\right)^{m_1-m_2}.
\end{equation}
Therefore, for a given $\eta>0$ the exists a number $N_\eta$ such that for any $n\geq N_\eta\geq  m_2(q-1)+m_1$, it is
\begin{equation}
\frac{c_{m_2 n}^{ {q}}}{c_{m_2 (m_1-m_2+n)}^{ {q}}} \leq \left(\frac{{q}}{{q}-1}\right)^{m_1-m_2} +\eta.
\end{equation}
It then follows, using \eqref{eq:61} and the fact $\Tr(\mathcal{N}(\vert m_1\rangle\langle m_1 \vert))=1$, that
\begin{equation}\label{eq:65}
\sum_{n =  N_\eta}^{\infty} c_{m_1 n}^{ {q}}\ln\frac{c_{m_2 n}^{ {q}}}{c_{m_2 (m_1-m_2+n)}^{ {q}}}\leq   \ln\left\{\left(\frac{{q}}{{q}-1}\right)^{m_1-m_2} +\eta\right\}  \sum_{n=0}^{\infty} c_{m_1 n}\leq    \ln\left\{\left(\frac{{q}}{{q}-1}\right)^{m_1-m_2} +\eta\right\}.
\end{equation}
Using relations \eqref{eq:66-1} and \eqref{eq:65}, we can get
\begin{eqnarray}
 \sum_{n =  N_\eta}^{\infty} c_{m_1 n}^{  q'}\ln\frac{c_{m_2 n}^{  q'}}{c_{m_2 (m_1-m_2+n)}^{  q'}} &\leq & \sum_{n =  N_\eta}^{\infty} c_{m_1 n}^{  q'}\ln\frac{c_{m_2 n}^{  q}}{c_{m_2 (m_1-m_2+n)}^{  q}} \\
   &\leq & \sum_{n =  N_\eta}^{\infty} c_{m_1 n}^{  q'}\ln\frac{c_{m_2 n}^{  q}}{c_{m_2 (m_1-m_2+n)}^{  q}}\\
    &\leq & \ln\left\{\left(\frac{ q}{ q-1}\right)^{m_1-m_2} +\eta\right\}  \sum_{n=N_{\eta}}^{\infty} c_{m_1 n}^{q'}   \\
    &\leq &\ln\left\{\left(\frac{ q}{ q-1}\right)^{m_1-m_2} +\eta\right\} \label{eq:73} .
\end{eqnarray}

Finally, we find that Eq. \eqref{eq:61} holds for $q' \leq \Theta_1+\Theta_2$, where
\begin{align}
  \Theta_1&\equiv
  m_1 \vert r\vert^2e^{-2m_1 s} \ln\frac{m_2 \vert r\vert^2e^{-2m_2 s} }{c_{m_2 (m_1 -m_2-1)}^{{q'}} }+\sum_{n=0, n\neq n'}^{N_\eta-1} c_{m_1 n}^{{q'}}\ln\frac{c_{m_2 n}^{{q'}}}{c_{m_2 (m_1 -m_2+n)}^{{q'}}}
  +\ln\left\{\left(\frac{ q}{ q-1}\right)^{m_1-m_2} +\eta\right\},  \label{eq:198}   \\
\Theta_2&\equiv  c_{m_1 n'}^{{q'}}\ln\frac{c_{m_2 n'}^{{q'}}}{c_{m_2 (m_1 -m_2+n')}^{{q'}}}\label{eq:199}.
\end{align}
Now, when $\varepsilon$ goes to $0$, the quantity $\Theta_1$ will remain finite  (it is  continuous with respect to $q$), while the quantity $\Theta_2$ diverges to $-\infty$. Therefore, for any rational number $q$ we can find a set $(q,q+\epsilon)$, for $q'$, which violates \eqref{eq:61}.
Since the set of rational numbers is dense into the set of reals, the proof follows.
\end{proof}


\section{Bounds on capacities uncertainty}
\label{app:OMB}

In this appendix we derive,  on the basis of relations \eqref{eq:lbchiN} and \eqref{eq:lbchiM},
tighter lower bounds on the sum $\chi_{H\ox}  + \chi_{A\ox}$ than the one from Theorem \ref{thm:uncertainty} for one-mode Gaussian channels. The reasoning is
based on the classification of OMG channels given in \cite{H07},
and the bounds are derived by using coherent states encoding.

\medskip

\begin{itemize}
\item
\textbf{Class A1:} $\boldsymbol{M}=0$, $\boldsymbol{N} \boldsymbol{N}^\top =\boldsymbol{I}$.

We have
\begin{equation}
\chi_{H\ox} \geq S\left(\frac{1}{2}\boldsymbol{I}\right)- S\left(\frac{1}{2}\boldsymbol{I}\right) = 0 ,
\end{equation}
and
\begin{equation}
\chi_{A\ox} \geq S\left(\left(P_E+\frac{1}{2}\right) \boldsymbol{I}\right)
- S\left(\frac{1}{2}\boldsymbol{I}\right) \geq g\left(P_E+\frac{1}{2}\right).
\end{equation}
Therefore, we get
\begin{equation}
\chi_{H\ox}  + \chi_{A\ox} \geq g\left(P_E+\frac{1}{2}\right).
\end{equation}

\medskip

\item
\textbf{Class A2:} $\boldsymbol{M}=\left(
                       \begin{array}{cc}
                         1 & 0 \\
                         0 & 0 \\
                       \end{array}
                     \right)$, $\boldsymbol{N} \boldsymbol{N}^\top =\boldsymbol{I}$.

We have
\begin{equation}
  \chi_{H\ox} \geq S\left(\left(
                       \begin{array}{cc}
                         P_A+1 & 0 \\
                         0 & \frac{1}{2} \\
                       \end{array}
                     \right)\right)- S\left(\left(
                       \begin{array}{cc}
                         1 & 0 \\
                         0 & \frac{1}{2} \\
                       \end{array}
                     \right)\right) ,
\end{equation}
and
\begin{equation}
  \chi_{A\ox} \geq S\left(\left(
                       \begin{array}{cc}
                         P_E+1 & 0 \\
                         0 & P_E +\frac{1}{2}\\
                       \end{array}
                     \right)\right)- S\left(\left(
                       \begin{array}{cc}
                         1 & 0 \\
                         0 & \frac{1}{2} \\
                       \end{array}
                     \right)\right) .
 \end{equation}
Therefore, we get
\begin{equation}
\chi_{H\ox}  + \chi_{A\ox}
\geq g\left( \sqrt{\left(P_A+1 \right)\frac{1}{2}} \right) + g\left( \sqrt{\left(P_E+1 \right)\left(P_E+\frac{1}{2}\right)} \right) - 2g\left(\sqrt{\frac{1}{2}} \right).
\end{equation}

\medskip

\item
\textbf{Class B1:} $\boldsymbol{M}=\boldsymbol{I}$,
$\boldsymbol{N} \boldsymbol{N}^\top =\frac{1}{2N_0+1}  \left(
                       \begin{array}{cc}
                         1 & 0 \\
                         0 & 0 \\
                       \end{array}
                     \right)  $.

We have
\begin{equation}
  \chi_{H\ox}\geq S \left(\left(
                       \begin{array}{cc}
                         P_A+\frac{1}{2}+\frac{1}{2N_0+1} & 0 \\
                         0 &P_A+\frac{1}{2}\\
                       \end{array}
                     \right)\right) - S \left(\left(
                       \begin{array}{cc}
                         \frac{1}{2}+\frac{1}{ 4N_0+2} & 0 \\
                         0 &\frac{1}{2}\\
                       \end{array}
                     \right)\right) ,
\end{equation}
and
\begin{equation}
  \chi_{A\ox}\geq S \left(\left(
                       \begin{array}{cc}
                         \frac{P_E+\frac{1}{2}}{2N_0+1}+\frac{1}{2} & 0 \\
                         0 &\frac{1}{2}\\
                       \end{array}
                     \right)\right) - S \left(\left(
                       \begin{array}{cc}
                         \frac{1}{2}+\frac{1}{ 4N_0+2} & 0 \\
                         0 &\frac{1}{2}\\
                       \end{array}
                     \right)\right) .
\end{equation}
Therefore, we get
\begin{align}
  \chi_{H\ox}  + \chi_{A\ox}
     &\geq  g\left( \sqrt{ \left(P_A+\frac{1}{2}+\frac{1}{2N_0+1}\right)\left( P_A+\frac{1}{2}  \right) }  \right) \notag\\
     &+g\left( \sqrt{ \frac{P_E+\frac{1}{2}}{4 N_0+2} +\frac{1}{2}  } \right) \notag\\
     &-2 g\left(\sqrt{\frac{1}{4}+ \frac{1}{4N_0+2} }        \right) .
\end{align}

\medskip

\item
\textbf{Class B2:} $\boldsymbol{M}=\boldsymbol{I}$,
$\boldsymbol{N} \boldsymbol{N}^\top =\frac{N_0}{N_0+\frac{1}{2}}\boldsymbol{I}$.

We have
\begin{equation}
\chi_{H\ox}  \geq S\left(  \left(P_A+\frac{1}{2}+\frac{N_0}{2N_0+1}\right)
\boldsymbol{I}    \right)- S\left(\left(\frac{1}{2}+ \frac{N_0}{2N_0+1}\right)\boldsymbol{I}          \right) ,
\end{equation}
and
\begin{equation}
\chi_{A\ox} \geq S\left(  \left(\frac{(P_E+\frac{1}{2}) N_0}{N_0+\frac{1}{2}}+\frac{1}{2}\right)\boldsymbol{I}    \right)- S\left(\left(\frac{1}{2}+ \frac{N_0}{2N_0+1}\right)\boldsymbol{I}          \right).
\end{equation}
Therefore, we get
\begin{equation}
\chi_{H\ox}  + \chi_{A\ox}
\geq g\left(  P_A+\frac{1}{2}+\frac{N_0}{2N_0+1}    \right)+g\left(  \frac{(P_E+\frac{1}{2}) N_0}{N_0+\frac{1}{2}}+\frac{1}{2}   \right)-2 g\left(\frac{1}{2}+ \frac{N_0}{2N_0+1}         \right) .
\end{equation}

\medskip

\item
\textbf{Class C Att:} $\boldsymbol{M}=\sqrt \kappa \boldsymbol{I}$,
$\boldsymbol{N}^\top \boldsymbol{N} =(1- \kappa)\boldsymbol{I}$,  $0<\kappa<1$

We have
\begin{equation}
\chi_{H\ox}  \geq S\left(  \left(\left(P_A+\frac{1}{2}\right)\kappa+ 1- \kappa\right)
\boldsymbol{I}    \right)- S\left(\frac{1}{2} \boldsymbol{I} \right) ,
\end{equation}
and
\begin{equation}
\chi_{A\ox} \geq S\left(  \left(\left(P_E+\frac{1}{2}\right)(1- \kappa)+ \kappa\right)\boldsymbol{I}    \right)- S\left(\frac{1}{2}\boldsymbol{I}\right).
\end{equation}
Therefore, we get
\begin{equation}\label{eq:boundCatt}
\chi_{H\ox}  + \chi_{A\ox} \geq  g\left(  \left(P_A+\frac{1}{2}\right) \kappa + 1- \kappa    \right)+ g\left(  \left(P_E+\frac{1}{2}\right)(1-\kappa)+ \kappa   \right).
\end{equation}

\medskip

\item
\textbf{Class C Amp:} $\boldsymbol{M}=\sqrt{\kappa}\boldsymbol{I}$,
$\boldsymbol{N} \boldsymbol{N}^\top =(\kappa-1)\boldsymbol{I}$, for $\kappa > 1$.

We have
\begin{equation}
\chi_{H\ox} \geq S\left(  \left(\left(P_A+\frac{1}{2}\right)\kappa+ \kappa-1\right)\boldsymbol{I}    \right)- S\left(\left(\kappa-\frac{1}{2}\right)\boldsymbol{I} \right) ,
\end{equation}
and
\begin{equation}
\chi_{A\ox} \geq S\left(  \left(\left(P_E+\frac{1}{2}\right)(\kappa-1)+\kappa\right)\boldsymbol{I}    \right)- S\left(\left(\kappa-\frac{1}{2}\right)\boldsymbol{I} \right).
\end{equation}
Therefore, we get
\begin{equation}\label{eq:boundCampl}
\chi_{H\ox}  + \chi_{A\ox}
\geq  g\left( \left(P_A+\frac{1}{2}\right)\kappa+ \kappa-1 \right)
      +g\left( \left(P_E+\frac{1}{2}\right)(\kappa-1)+\kappa \right) - 2 g\left(\kappa-\frac{1}{2}\right).
 \end{equation}
\medskip

\item
\textbf{Class D:} $\boldsymbol{M}=\sqrt{-\kappa}\boldsymbol{Z}$,
$\boldsymbol{N}^\top \boldsymbol{N} =(1-\kappa)\boldsymbol{I}$, $\kappa\in(-\infty, 0)$

We have
\begin{equation}
\chi_{H\ox} \geq S\left(  \left(\left( P_A+\frac{1}{2}\right)\vert\kappa\vert+\vert 1-\kappa\vert\right)
\boldsymbol{I}    \right)- S\left(\frac{\vert\kappa\vert+\vert 1-\kappa\vert}{2}\boldsymbol{I}         \right) .
\end{equation}
and
\begin{equation}
\chi_{A\ox} \geq S\left( \left(\left( P_E+\frac{1}{2}\right)\left(\vert 1-\kappa\vert\right)+\vert\kappa\vert\right)
\boldsymbol{I}    \right)- S\left(\frac{\vert\kappa\vert+\vert 1-\kappa\vert}{2}\boldsymbol{I}    \right).
\end{equation}
Therefore, we obtain
\begin{equation}
\chi_{H\ox}  + \chi_{A\ox} \geq   g\left(   \left(P_A+\frac{1}{2}\right)\vert\kappa\vert +1-\kappa   \right)+ g\left(  \left( P_E+\frac{1}{2}\right)( 1-\kappa)+\vert\kappa\vert    \right)- 2g\left(\frac{\vert\kappa\vert+\vert 1-\kappa\vert}{2}   \right).
\end{equation}
\end{itemize}

\begin{remark}
For an easy comparison with the bound in Theorem \ref{thm:uncertainty}, let us consider the class C.
The r.h.s. of \eqref{eq:boundCatt} and \eqref{eq:boundCampl} can be put together as
\begin{equation}\label{eq:boundtog}
  g\left( \left(P_A+\frac{1}{2}\right)\kappa+|1-\kappa| \right)
    +g\left(  \left(P_E+\frac{1}{2}\right)|1-\kappa|+\kappa \right) - 2 g\left(\frac{|1-\kappa|+\kappa}{2}\right).
\end{equation}
Due to the properties of the function $g$ defined in \eqref{eq:gfunction}, it is
\begin{align}\label{eq:newboundtog}
\text{Eq.\eqref{eq:boundtog}} &\geq g\left(  \left(\min\left\{P_A,P_E\right\}+\frac{1}{2}\right)\kappa+|1- \kappa|    \right)+g\left(  \left(\min\left\{P_A,P_E\right\}+\frac{1}{2}\right)|1-\kappa|+\kappa    \right) \notag\\
&-2 g\left(\frac{|1-\kappa|+\kappa}{2}\right).
\end{align}
Still referring to the properties of the function $g$, we have that the quantity \eqref{eq:newboundtog}
grows, in terms of $\min\left\{P_A,P_E\right\}$, faster than \eqref{eq:uncertainty}. Thus, the minimum difference between the two bounds (\eqref{eq:boundtog} and \eqref{eq:uncertainty}) is achieved when
$\min\left\{P_A,P_E\right\}$ goes to zero and results
at least as much big as
\begin{align}
&g\left(1-\frac{1}{2}\kappa\right)+g\left(\frac{1+\kappa}{2}\right), \qquad \kappa<1 \\
&2g\left(\frac{3}{2}\kappa-1\right)-2g\left(\kappa-\frac{1}{2}\right),\qquad \kappa>1.
\end{align}
\end{remark}
These two quantities being positive, this shows the tightness of
\eqref{eq:boundtog} with respect to \eqref{eq:uncertainty}.


\end{document}